\documentclass[11pt,a4paper]{article} 
\usepackage[applemac]{inputenc}
\usepackage[english]{babel}
\usepackage{pdfsync}
\usepackage{amsmath,amsthm,amssymb,cmll,bussproofs,pst-all,pstricks-add,mathtools}
\usepackage{verbatim}

\theoremstyle{plain}
\newtheorem{theorem}{Theorem}
\newtheorem{corollary}{Corollary}[theorem]
\newtheorem{proposition}[theorem]{Proposition}
\newtheorem{lemma}[theorem]{Lemma}

\newtheorem*{Atheorem}{Theorem}
\newtheorem*{Aproposition}{Proposition}

\theoremstyle{definition}
\newtheorem{definition}[theorem]{Definition}
\newtheorem*{Adefinition}{Definition}

\theoremstyle{remark}
\newtheorem*{remark}{Remark}

\newcommand{\pol}{\simbot} 
\newcommand{\poll}{\simperp} 
\newcommand{\cond}[1]{\mathbf{#1}} 
\newcommand{\de}[1]{\mathfrak{#1}} 
\newcommand{\sca}[2]{\mathopen{\ll} \de{#1},\de{#2} \mathclose{\gg}} 
\newcommand{\hil}[1]{\mathbb{#1}} 
\newcommand{\vn}[1]{\mathcal{#1}} 
\newcommand{\finhyp}{\vn{R}} 
\newcommand{\infhyp }{\vn{R}_{0,1}} 
\newcommand{\vntimes}{\otimes} 

\newcommand{\what}[1]{~\widehat{\!\!#1\!\!}~}
\newcommand{\plug}{\mathop{\dblcolon}} 
\newcommand{\path}[2]{\mathit{Path}^{#1}(#2)} 
\newcommand{\cycl}[1]{\mathcal{C}(#1)} 
\newcommand{\disjun}{\uplus} 
\newcommand{\bicol}{\square} 
\newcommand{\catmll}{$\mathfrak{Graph}_{MLL}$} 
\newcommand{\scalar}[1]{\mathopen{\ll} #1 \mathclose{\gg}} 
\newcommand{\wgraph}[1]{#1=(V_{#1},E_{#1},s_{#1},t_{#1},\omega_{#1})}   
\newcommand{\mat}[1]{\mathcal{M}_{#1}}   
\newcommand{\norm}[1]{\mathopen{\lVert}#1\mathclose{\rVert}} 
\newcommand{\<}{\mathopen{<}}
\renewcommand{\>}{\mathclose{>}}

\newcommand{\abs}[1]{\mathopen{|}#1\mathclose{|}} 

\usepackage{color}

\title{Interaction Graphs: Multiplicatives}
\author{Thomas Seiller \\ \small{\begin{tabular}{cc} IML — UMR 6206 & LAMA — UMR 5127\\ CNRS - Université Aix-Marseille & CNRS - Université de Savoie\\ 163 Avenue de Luminy, Case 907 & Bât. Chablais, Campus Scientifique\\ 13288 Marseille Cedex 09 & 73376 Le Bourget-du-Lac Cedex\\ France & France\end{tabular}}\\\textbf{seiller@iml.univ-mrs.fr}}
\date{January 28th, 2011}

\begin{document}
\maketitle
\thispagestyle{empty}

\begin{abstract}
We introduce a graph-theoretical representation of proofs of multiplicative linear logic which yields both a denotational semantics and a notion of truth. For this, we use a locative approach (in the sense of ludics \cite{locussolum}) related to game semantics \cite{hylandong,abramsky94full} and the Danos-Regnier interpretation of GoI operators as paths in proof nets \cite{AspertiDanosLaneveRegnier94,DanosRegnier95}. We show how we can retrieve from this locative framework both a categorical semantics for MLL with distinct units and a notion of truth. Moreover, we show how a restricted version of our model can be reformulated in the exact same terms as Girard’s geometry of interaction \cite{goi5}. This shows that this restriction of our framework gives a combinatorial approach to J.-Y. Girard’s geometry of interaction in the hyperﬁnite factor, while using only graph-theoretical notions. 
\end{abstract}

\section*{Introduction}

We develop a graph-theoretical geometry of interaction multiplicative linear logic which yields both a denotational semantics and a notion of truth, and draws bridges between game semantics and the latest developments in geometry of interaction. This work is inspired by Girard's latest paper on geometry of interaction \cite{goi5} and uses one of its key ideas: “locativity", which first appeared in Ludics \cite{locussolum}. Locative semantics can be considered as a geometrical implementation of denotational semantics, where the objects interpreting the proofs have a precise location, just as computer programs and data have a physical address in memory. This physical implementation has a peculiarity that could be seen as a drawback: we define only partial operations (for instance, the tensor product is defined only when the locations are disjoint) and some properties need additional hypotheses (for instance, the associativity of reduction). However, we will show that this does not lessen in any way our interpretation since we can define, working “modulo delocations" (which can be seen as internal isomorphisms), a $\ast$-autonomous category, yielding a denotational semantics of MLL. From our locative framework, we can also define a notion of success, which would correspond to the game semantics' notion of winning. This allows us to define a notion of truth which is consistent and preserved by composition.

The last part of the paper is devoted to a linear-algebraic reformulation of our semantics, when restricted to a certain class of objects. We show how the notions we introduced on graphs corresponds to the definitions found in Girard's latest paper \cite{goi5}. Our work can therefore be seen as both an operator-free finite-dimensional introduction to this latest geometry of interaction and an explanation of some of its peculiarities by making explicit its relations to previous works on geometry of interaction \cite{multiplicatives,goi1} and proof nets \cite{linearlogic,DanosRegnier95,AspertiDanosLaneveRegnier94}.

\subsubsection*{Geometry of Interaction in the Hyperfinite Factor}

The geometry of interaction program was introduced \cite{towards} by Jean-Yves Girard soon after the introduction of linear logic \cite{linearlogic}. It aims at giving a semantics of cut-elimination by representing proofs as operators. Several versions of geometry of interaction were introduced by Girard, all using the tools of operator theory (a good introduction to the theory of operator algebras can be found in Murphy's book \cite{murphy}). These versions of geometry of interaction use two key ingredients: 
\begin{itemize}
\item the operators which represent proofs 
\item the notion of interaction, representing cut elimination
\end{itemize}

The latest version of geometry of interaction Girard introduced (which we will call GoI5 \cite{goi5}) uses advanced operator-theoretic notions. It offers great flexibility in its definition of exponentials and is therefore particularly promising when it comes to the study of complexity. Moreover, its use of operator algebras and its close relation to quantum coherent spaces \cite{QCS} suggest future applications to quantum computing.

It presents severals differences with the preceding versions, and most of these can be found in our graph-theoretic framework. The first important thing is that the considered set of operators is not limited to partial isometries. The second is that the adjunction, which relates the tensor product and the linear implication, is given by
\begin{equation}
ldet(1-F.(A+B))=ldet(1-F.A)+ldet(1-[F]A.B)\nonumber
\end{equation}
Indeed, usual denotational semantics adjunctions state an equality between two quantities: for instance, in coherent semantics the adjunction is given by $\sharp ([F]a\cap b)=\sharp (F\cap(a\times b))$. Thus the adjuction of the latest geometry of interaction differs from what we are used to because of the additional term $ldet(1-F.A)$. The presence of this additional term is compensated by the use of the so-called wager, a real number that can be considered as a sort of truth-value (actually, the cologarithm of a truth value) where $0$ means true, and $\infty$ means false. 

\subsubsection*{Locativity}

As in ludics \cite{locussolum}, proofs — hence formulas, which are defined as sets of proofs — have a definite location. In ludics, locations were defined as a finite sequence of integers (the locus), while in GoI5 the location is given by a finite projection in the hyperfinite factor of type $\text{II}_{\infty}$. It means, in particular, that if $a$ and $b$ are two objects (representing proofs) and $\phi(a),\psi(b)$ are isomorphic copies on different locations, the execution — which corresponds to cut-elimination — $a\plug b$ need not be (and will not be in general) isomorphic to $\phi(a)\plug\psi(b)$.

While it can be argued that geometry of interaction always had a locative flavor, it is only in this latest version of it that it becomes fully explicit. Indeed, making the choice of a fully locative framework allows one to dispense from the use of the partial isometries $p$ and $q$ that were omnipresent in the first versions of geometry of interaction \cite{goi1}. We chose, in this paper, to adopt a locative framework, although the choice of following the more ancient versions of geometry of interaction would also have been a valid one. We made this choice for the sake of simplicity: indeed, even if locativity makes the interpretation of proofs and formulas more complicated, it actually makes both the definitions of low-level operations (defined on graphs in our setting) and the embedding of our framework in Girard's GoI5 simpler.

\section{Some Results on Graphs}\label{graphssection}

In the following, we will work with directed weighted graphs. The use of graphs is reminiscent of many other related works, e.g. Kelly-MacLane graphs \cite{KellyMcLane}, linear logic proof nets \cite{linearlogic}, sharing graphs \cite{AbadiGonthierLevy92a,DanosRegnier95}, etc. However, our approach differs from these in many ways, among them the fact that our framework is locative, that we work with weighted (non simple) graphs, and the possibility of defining a notion of truth.

\begin{definition}
A \emph{directed weighted graph} is a tuple $G=(V_{G},E_{G},s_{G},t_{G},\omega_{G})$, where $V_{G}$ is the set of vertices, $E_{G}$ is the set of edges, $s_{G}$ and $t_{G}$ are two functions from $E_{G}$ to $V_{G}$, the \emph{source} and \emph{target} functions, and $\omega_{G}$ is a function $E_{G} \rightarrow ]0,1]$.
\end{definition}

In this paper, all the directed weighted graphs considered will have a \emph{finite} set of vertices, and a \emph{finite or countably infinite} set of edges.

We will write $E_{G}(v,w)$ for the set of all edges $e\in E_{G}$ satisfying $s_{G}(e)=v$ and $t_{G}(e)=w$. Moreover, we will sometimes forget the subscripts when the context is clear.

\begin{definition}[Simple graphs]
We say a directed weighted graph $G$ is \emph{simple} when there is no more than one edge between two given vertices.
\end{definition}

\begin{definition}
From a directed weighted graph $G$, we can define a simple graph $\what{G}$ with weights in $\mathbb{R}_{>0}\cup\{\infty\}$:
\begin{eqnarray*}
V_{\what{G}}&=&V_{G}\\
E_{\what{G}}&=&\{(v,w)~|~\exists e\in E_{G}, s_{G}(e)=v, t_{G}(e)=w\}\\
\omega_{\what{G}}&:&(v,w)\mapsto\sum_{e\in E_{G}(v,w)}\omega_{G}(e)
\end{eqnarray*}
When the weights of $\what{G}$ are in $\mathbb{R}_{>0}$, we say $\what{G}$ is \emph{total}.
\end{definition}

We will now define a construction on graphs that will allow us to consider, given two graphs $G$ and $H$, paths that alternate between an edge in $G$ and an edge in $H$, a construction that is quite standard in the literature \cite{abramhaghverdiscott,abramsky94full,defalco}. The first construction (the plugging of two graphs) is the keystone around which this paper is constructed. Once we can talk of alternating paths and cycles, we will be able to obtain the two main results: a reduction operation that is associative — which corresponds to cut-elimination, and a three-term equality (Proposition \ref{adjunct}) from which we will be able to define our adjunction.

We will denote by $\disjun$ the disjoint union on sets. Given sets $E,F$ and $X$ and two functions $f: E \rightarrow X$ and $g: F\rightarrow X$, we will write $f\disjun g$ the function from $E\disjun F$ to $X$ that is defined by the universality property of coproducts, i.e. the “co-pairing" of $f$ and $g$.

\begin{definition}[Union of graphs]
Given two graphs $G$ and $H$, we can define the union graph $G\cup H$ of $G$ and $H$ as:
\begin{equation}
(V_{G}\cup V_{H},E_{G}\disjun E_{H},(\iota_{G}\circ s_{G})\disjun (\iota_{H}\circ s_{H}),(\iota_{G}\circ t_{G})\disjun (\iota_{H}\circ t_{H}),\omega_{G}\disjun\omega_{H})\nonumber
\end{equation}
where $\iota_{G}$ (resp. $\iota_{H}$) denotes the inclusion of $V^{G}$ (resp. $V^{H}$) in $V^{G}\cup V^{H}$.
\end{definition}

\begin{remark}
We want to stress the fact that while we take a disjoint union of the edges, we take the union of the sets of vertices. Therefore, the union of two graphs may not be equal to the union of two isomorphic copies, since nothing tells us that these isomorphic copies will intersect over the same (up to the isomorphism) set of vertices. This is where locativity takes all its importance, since a non-empty intersection of the sets of vertices of two graphs is the place where the interaction will occur.
\end{remark}

Now, in order to consider paths that alternate between two graphs, we need to keep track of the origin of the edges, which motivates the following definition.

\begin{definition}[Plugging]
Given two graphs $G$ and $H$, we define the graph $G\bicol H$ as the union graph of $G$ and $H$, together with a coloring function $\delta$ from $E_{G}\disjun E_{H}$ to $\{0,1\}$ such that 
\begin{equation}
\left\{\begin{array}{l}
\delta(x)=0\text{ if }x\in E_{G}\\
\delta(x)=1\text{ if }x\in E_{H}
\end{array}\right.
\nonumber
\end{equation}
We refer to $G\bicol H$ as the \emph{plugging of $G$ and $H$.}
\end{definition}
Figure \ref{plug} shows an example of the plugging of the graphs $F$ and $G$ from Figure \ref{graphs} in which colors are represented by the location of the edges: the top edges are the $0$-colored edges, while the bottom edges are the $1$-colored ones.

\begin{figure}[b]
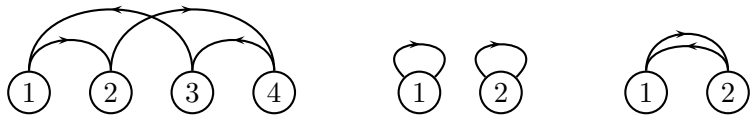

\begin{center}
$\begin{array}{ccccc}
~\\
\psmatrix[colsep=0.5cm,rowsep=0.5cm,mnode=circle]
1&2&3&4
\psset{ArrowInside=->}
\nccurve[angleA=90,angleB=90,ncurvA=1,ncurvB=1]{1,1}{1,2}
\nccurve[angleA=90,angleB=90,ncurvA=1,ncurvB=1]{1,2}{1,4}
\nccurve[angleA=90,angleB=90,ncurvA=1,ncurvB=1]{1,4}{1,3}
\nccurve[angleA=90,angleB=90,ncurvA=1,ncurvB=1]{1,3}{1,1}
\endpsmatrix
&~~~~~&
\psmatrix[colsep=0.5cm,rowsep=0.5cm,mnode=circle]
1&2
\psset{ArrowInside=->}
\nccurve[angleA=135,angleB=45,ncurvA=4,ncurvB=4]{1,1}{1,1}
\nccurve[angleA=135,angleB=45,ncurvA=4,ncurvB=4]{1,2}{1,2}
\endpsmatrix
&~~~~~&
\psmatrix[colsep=0.5cm,rowsep=0.5cm,mnode=circle]
1&2
\psset{ArrowInside=->}
\nccurve[angleA=90,angleB=90,ncurvA=1.2,ncurvB=1.2]{1,1}{1,2}
\nccurve[angleA=90,angleB=90,ncurvA=0.8,ncurvB=0.8]{1,2}{1,1}
\endpsmatrix
\end{array}$
\end{center}
\caption{Graphs $F$, $G$ and $H$}\label{graphs}
\end{figure}

\begin{figure}
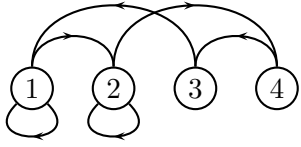

\begin{center}
$\begin{array}{ccc}
\psmatrix[colsep=0.5cm,rowsep=0.5cm,mnode=circle]
1&2&3&4
\psset{ArrowInside=->}
\nccurve[angleA=90,angleB=90,ncurvA=1,ncurvB=1]{1,1}{1,2}
\nccurve[angleA=90,angleB=90,ncurvA=1,ncurvB=1]{1,2}{1,4}
\nccurve[angleA=90,angleB=90,ncurvA=1,ncurvB=1]{1,4}{1,3}
\nccurve[angleA=90,angleB=90,ncurvA=1,ncurvB=1]{1,3}{1,1}
\nccurve[angleA=-45,angleB=-135,ncurvA=4,ncurvB=4]{1,1}{1,1}
\nccurve[angleA=-45,angleB=-135,ncurvA=4,ncurvB=4]{1,2}{1,2}
\endpsmatrix
\end{array}$
\end{center}
\caption{Plugging of $F$ and $G$}\label{plug}
\end{figure}

\begin{definition}[Paths, cycles and $k$-cycles]
A \emph{path} in a graph $G$ is a finite sequence of edges $(e_{i})_{0\leqslant i\leqslant n}\text{ }(n\in\mathbb{N})$ in $E_{G}$ such that $s(e_{i+1})=t(e_{i})$ for all $0\leqslant i\leqslant n-1$. We will call the vertices $s(\pi)=s(e_{0})$ and $t(\pi)=t(e_{n})$ the beginning and the end of the path.

We will also call a \emph{cycle} a path $\pi=(e_{i})_{0\leqslant i\leqslant n}$ such that $s(e_{0})=t(e_{n})$. If $\pi$ is a cycle, and $k$ is the greatest integer such that there exists a cycle $\rho$ with\footnote{Here, we denote by $\rho^{k}$ the concatenation of $k$ copies of $\rho$.} $\pi=\rho^{k}$, we will say that $\pi$ is a \emph{$k$-cycle}.
\end{definition}

\begin{proposition}\label{countingeqclas}
Let $\rho=(e_{i})_{0\leqslant i\leqslant n-1}$ be a cycle, and let $\sigma$ be the permutation taking $i$ to $i+1$ ($i=0,\dots,n-2$) and $n-1$ to $0$. We define the set
\begin{equation}
\bar{\rho}=\{(e_{\sigma^{k}(i)})_{0\leqslant i\leqslant n-1}~|~ 0\leqslant k\leqslant n-1\}\nonumber
\end{equation}
Then $\rho$ is a $k$-cycle if and only if the cardinality of $\bar{\rho}$ is equal to $n/k$. In the following, we will refer to such an equivalence class modulo cyclic permutations as a \emph{circuit}, or \emph{$k$-circuit}.
\end{proposition}

\begin{proof}
We use classical cyclic groups techniques here. We will abusively denote by $\sigma^{p}(\rho)$ the path $(e_{\sigma^{p}(i)})_{0\leqslant i\leqslant n-1}$.

First, notice that if $\rho$ is a $k$-cycle, then $\sigma^{n/k}(\rho)=\rho$. Now, if $s$ is the smallest integer such that $\sigma^{s}(\rho)=\rho$, we have that $e_{i+s}=e_{i}$. Hence, writing $m=n/s$, we have $\rho=\pi^{m}$ where $\pi=(e_{i})_{0\leqslant i\leqslant s-1}$. This implies that $k=n/s$ from the maximality of $k$. Hence $\rho$ is a $k$-cycle if and only if the smallest integer $s$ such that $\sigma^{s}(\rho)=\rho$ is equal to $n/k$.

Let $s$ be the smallest integer such that $\sigma^{s}(\rho)=\rho$. We have that for any integers $p,q$ such that $0\leqslant q< s$, $\sigma^{ps+q}(\rho)=\sigma^{q}(\rho)$. Indeed, it is a direct consequence of the fact that $\sigma^{ps}(\rho)=\rho$ for any integer $p$. Moreover, since $\sigma^{n}(\rho)=\rho$, we have that $s$ divides $n$. Hence, we have that the cardinality of $\bar{\rho}$ is at most $s$. To show that the cardinality of $\bar{\rho}$ is exactly $s$, we only need to show that $\sigma^{i}(\rho)\neq\sigma^{j}(\rho)$ for $i< j$ between $0$ and $s-1$. But if it were the case, we would have, since $\sigma$ is a bijection, $\rho=\sigma^{j-i}(\rho)$, an equality contradicting the minimality of $s$.
\end{proof}

\begin{definition}[Alternating paths]
Let $\wgraph{G}$ and $\wgraph{H}$ be two graphs. We define the \emph{alternating paths} between $G$ and $H$ as the paths $(e_{i})_{0\leqslant i\leqslant n}$ in $G\bicol H$ which satisfy
\begin{equation}
\delta(e_{i})\neq\delta(e_{i+1})~~~~(0\leqslant i\leqslant n-1)\nonumber
\end{equation}
We will denote by $\path{}{G,H}$ the set of alternating paths in the graph $G\bicol H$, and by $\path{v,w}{G,H}$ the set of alternating paths in $G\bicol H$ beginning at $v$ and ending at $w$. 
We will call an \emph{alternating cycle} in $G\bicol H$ a cycle $(e_{i})_{0\leqslant i\leqslant n}\in\path{}{G,H}$ such that $\delta(e_{n})\not=\delta(e_{0})$.
\end{definition}

\begin{remark} The last condition $\delta(e_{n})\neq\delta(e_{0})$ is necessary because we want to consider only the cycles that induce an infinite number of alternating paths in the execution, i.e. the paths that can be travelled through more than once. An alternating path $(e_{i})_{0\leqslant i\leqslant n}$ such that $\delta(e_{n})=\delta(e_{0})$ cannot be taken travelled through twice in a row since the path $(f_{i})_{0\leqslant i\leqslant 2n+1}$ defined as $f_{i}=f_{i+n+1}=e_{i}$ for $0\leqslant i\leqslant n$ is not alternating ($\delta(f_{n})=\delta(f_{n+1})$.
\end{remark}

\begin{remark} With the notations of the above definition, if $V_{G}\cap V_{H}=\emptyset$, the set of alternating paths $\path{}{G,H}$ is reduced to the set $E_{G}\disjun E_{H}$, modulo the identification between edges and paths of length $1$.
\end{remark}

\begin{definition}[The set of $1$-circuits]
We will denote by $\cycl{G,H}$ the set of alternating $1$-circuits in $G\bicol H$, i.e. the quotient of the set of alternating $1$-cycles by cyclic permutations.
\end{definition}

\begin{remark}
All cycles in $\cycl{G,H}$ are of even length, since they have to satisfy the condition $\delta(e_{0})\not=\delta(e_{n})$.
\end{remark}

We then extend the weight function to paths in the following way.

\begin{definition}
The weight of a path $\pi=(e_{i})_{0\leqslant i\leqslant n}$ in a weighted graph $G$ is defined by $\omega_{G}(\pi)=\prod_{i=0}^{n} \omega_{G}(e_{i})$.

This definition does not depend on the path but only on the set of the edges it is composed of. It is therefore invariant under cyclic permutations and we define the weight of a circuit as the weight of any cycle in the class.
\end{definition}

We define an operation on graphs, which we will call \emph{reduction}, which is again quite standard, and is a straightforward generalization of the execution formula between permutations \cite{multiplicatives,blindspot} which corresponds to cut-elimination in proof nets.
\begin{definition}[Reduction]~\\
Let $G=(V_{G}, E_{G},s_{G},t_{G},\omega_{G})$ and $H=(V_{H}, E_{H},s_{H},t_{H},\omega_{H})$ be two graphs. Denoting by $V_{G}\Delta V_{H}$ the symmetric difference of $G$ and $H$, we define \emph{the reduction of $G$ and $H$} as the graph $G\plug H$ defined by 
\begin{equation}
\begin{array}{c}
\begin{array}{ccc}V_{G\plug H}= V_{G}\Delta V_{H}&~~&E_{G\plug H}= \bigcup{}_{v,w\in V_{G\plug H}}\path{v,w}{G,H}\\
s_{G\plug H}:  (e_{i})_{0\leqslant i\leqslant n}\mapsto s_{G\bicol H}(e_{0})&~~&
t_{G\plug H}:  (e_{i})_{0\leqslant i\leqslant n}\mapsto t_{G\bicol H}(e_{n})\end{array}\\
\omega_{G\plug H}: (e_{i})_{0\leqslant i\leqslant n}\mapsto \omega_{G\bicol H}((e_{i})_{0\leqslant i\leqslant n})
\end{array}\nonumber
\end{equation}
\end{definition}

\begin{remark}
Notice that if the graphs $G,H$ have disjoint sets of vertices (i.e. if $V_{2}=\emptyset$), then $G\plug H$ is equal to $G\cup H$.
\end{remark}

\begin{remark}
The operation of reduction is similar to the “composition and hiding" of strategies in game semantics. Indeed, one can think of a directed graph $G=(V^{G},E^{G},s^{G},t^{G})$ as a non-deterministic strategy where the vertices represent the moves. From a graph $G$, one can define the graph $G^{\dagger}=(V^{G}\times\{s,t\},E^{G},\tilde{s}^{G},\tilde{t}^{G})$, where $\tilde{s}^{G}(e)=(s^{G}(e),s)$ and $\tilde{t}^{G}(e)=(t^{G}(e),t)$ ($s,t$ are used as polarities here, and could very well be named $+,-$). Then, considering the set of alternating paths in the graph $G\bicol H$ is the same as considering the set of paths in the graph $G^{\dagger}\bicol (H^{\dagger})^{\ast}$, where $(H^{\dagger})^{\ast}$ is the graph $H^{\dagger}$ where $s$ and $t$ have been interchanged (the change of polarity). Now, the composition and hiding of the two strategies $G^{\dagger}$ and $H^{\dagger}$ corresponds to taking the set of paths in $G^{\dagger}\bicol (H^{\dagger})^{\ast}$ whose sources and targets are in $(V^{G}\Delta V^{H})\times\{s,t\}$, i.e. the graph $(G\plug H)^{\dagger}$.
\end{remark}

Figure \ref{internal} shows the alternating paths in $F\bicol G$ and $F\bicol H$, where $F$, $G$ and $H$ are the graphs defined in Figure \ref{graphs}. Notice the internal cycle that appears between $F$ and $H$.

\begin{figure}
\begin{center}
$\begin{array}{ccc}
~\\
\psmatrix[colsep=0.5cm,rowsep=0.5cm,mnode=circle]
[linestyle=dashed]1&[linestyle=dashed]2&3&4
\psset{ArrowInside=->}
\nccurve[angleA=90,angleB=90,ncurvA=1,ncurvB=1,offsetA=10pt,offsetB=-10pt]{1,1}{1,2}
\nccurve[angleA=90,angleB=90,ncurvA=1,ncurvB=1,offsetA=-10pt]{1,2}{1,4}
\nccurve[angleA=90,angleB=90,ncurvA=1,ncurvB=1]{1,4}{1,3}
\nccurve[angleA=90,angleB=90,ncurvA=1,ncurvB=1,offsetB=10pt]{1,3}{1,1}
\nccurve[angleA=-90,angleB=-90,ncurvA=1,ncurvB=1,offsetA=-10pt,offsetB=-10pt]{1,2}{1,2}
\nccurve[angleA=-90,angleB=-90,ncurvA=1,ncurvB=1,offsetA=10pt,offsetB=10pt]{1,1}{1,1}
\nccurve[angleA=90,angleB=-90,linestyle=dashed,offsetA=-10pt,offsetB=-10pt]{1,1}{1,1}
\nccurve[angleA=-90,angleB=90,linestyle=dashed,offsetA=-10pt,offsetB=-10pt]{1,1}{1,1}
\nccurve[angleA=90,angleB=-90,linestyle=dashed,offsetA=10pt,offsetB=10pt]{1,2}{1,2}
\nccurve[angleA=-90,angleB=90,linestyle=dashed,offsetA=10pt,offsetB=10pt]{1,2}{1,2}
\endpsmatrix
&~~~~~~~~~&
\psmatrix[colsep=0.5cm,rowsep=0.5cm,mnode=circle]
[linestyle=dashed]1&[linestyle=dashed]2&3&4
\psset{ArrowInside=->}
\nccurve[angleA=90,angleB=90,ncurvA=1,ncurvB=1,offsetA=-10pt,offsetB=-10pt]{1,1}{1,2}
\nccurve[angleA=90,angleB=90,ncurvA=1,ncurvB=1,offsetA=-10pt]{1,2}{1,4}
\nccurve[angleA=90,angleB=90,ncurvA=1,ncurvB=1]{1,4}{1,3}
\nccurve[angleA=90,angleB=90,ncurvA=1,ncurvB=1,offsetB=-10pt]{1,3}{1,1}
\nccurve[angleA=-90,angleB=-90,ncurvA=1,ncurvB=1,offsetA=-10pt,offsetB=-10pt]{1,1}{1,2}
\nccurve[angleA=-90,angleB=-90,ncurvA=1,ncurvB=1,offsetA=-10pt,offsetB=-10pt]{1,2}{1,1}
\nccurve[angleA=90,angleB=-90,linestyle=dashed,offsetA=10pt,offsetB=10pt]{1,1}{1,1}
\nccurve[angleA=-90,angleB=90,linestyle=dashed,offsetA=10pt,offsetB=10pt]{1,1}{1,1}
\nccurve[angleA=90,angleB=-90,linestyle=dashed,offsetA=10pt,offsetB=10pt]{1,2}{1,2}
\nccurve[angleA=-90,angleB=90,linestyle=dashed,offsetA=10pt,offsetB=10pt]{1,2}{1,2}
\endpsmatrix
\end{array}$
\end{center}
\caption{Alternating paths in $F\bicol G$ and $F\bicol H$}\label{internal}
\end{figure}

\begin{proposition}[Associativity]\label{assocmll}
Let $G_{i}=(V_{i},E_{i},s_{i},t_{i},\omega_{i})~(i=0,1,2)$ be three graphs with $V_{0}\cap V_{1}\cap V_{2}=\emptyset$. We have:
\begin{equation}
G_{0}\plug (G_{1}\plug G_{2})=(G_{0}\plug G_{1})\plug G_{2}\nonumber
\end{equation}
\end{proposition}

\begin{proof}
Let us define the 3-colored graph $G_{0}\square G_{1}\square G_{2}$ as the union graph $(\bigcup V_{i},\biguplus E_{i},\biguplus s_{i}, \biguplus t_{i})$ together with the coloring function $\delta$ from $\biguplus E_{i}$ to $\{0,1,2\}$ which associates to each edge the number $i$ of the graph $G_{i}$ it comes from. We consider the 3-alternating paths between $G_{0},G_{1},G_{2}$, that is the paths $(e_{i})$ in $G_{0}\square G_{1}\square G_{2}$ satisfying:
\begin{equation}
\delta(e_{i})\neq\delta(e_{i+1})\nonumber
\end{equation}
Then, we can define the simultaneous reduction of $G_{0},G_{1},G_{2}$ as the graph $\mathop{\dblcolon}_{i}G_{i}=(V_{0}\Delta V_{1}\Delta V_{2},F,s_{F},t_{F})$, where $F$ is the set of 3-alternating paths between $G_{0},G_{1},G_{2}$, $s_{F}(e)$ is the beginning of the path $e$ and $t_{F}(e)$ is its end.

We then show that this induced graph $\mathop{\dblcolon}_{i}G_{i}$ is equal to $(G_{0}\plug G_{1})\plug G_{2}$ and $G_{0}\plug (G_{1}\plug G_{2})$. This is a simple verification. Indeed, to prove for instance that $\mathop{\dblcolon}_{i}G_{i}$ is equal to $(G_{0}\plug G_{1})\plug G_{2}$, we just write the 3-alternating paths in $G_{0},G_{1},G_{2}$ as an alternating sequence of alternating paths in $G_{0}\bicol G_{1}$ (with\footnote{This is where the hypothesis $V_{0}\cap V_{1}\cap V_{2}=\emptyset$ is important. If this is not satisfied, one gets some 3-alternating paths of the form $\rho x$, where $x$ is an edge in $G_{2}$ and $\rho$ is an alternating path in $G_{0}\bicol G_{1}$, but such that $\rho$ does not correspond to an edge in $G_{0}\plug G_{1}$.} source and target in $V_{0}\Delta V_{1}$, i.e. an edge of $G_{0}\plug G_{1}$) and edges in $G_{2}$.
\end{proof}

\begin{remark}
Notice that reduction is not a composition of functions, and, because of the locativity of our framework, associativity is true only under an additional assumption on how the three graphs intersect. To get a counter-example, just take three graphs $F,G,H$ with $V_{F}=V_{G}=V_{H}=\{1\}$ such that $F,G$ have no edges and $H$ has only one edge (necessarily of source and target $1$): then $F\plug (G\plug H)=F$ and $(F\plug G)\plug H=H$.

However, we will get a genuine associativity when defining our category, since composition will be defined up to delocation (see section \ref{denot}).
\end{remark}

We then get the following proposition that will allow us to define our three-term adjunction.

\begin{proposition}\label{adjunct}
Let $G$, $H$ and $F$ be directed graphs, with $V_{G}\cap V_{H}=\emptyset$ and $V_{G}\cup V_{H}\subseteq V_{F}$. We have, denoting by $\sharp A$ the cardinality of $A$,
\begin{equation}
\sharp(\cycl{F,G\cup H})=\sharp(\cycl{F,G})+\sharp(\cycl{F\plug G,H})\nonumber
\end{equation}
\end{proposition}

\begin{proof}
Given an alternating $1$-circuit $\{(e_{i})\}_{0\leqslant 2n-1}$ between $G\cup H$ and $F$, we have two cases. First, the $1$-circuit can be a sequence of edges between vertices in $V_{G}$, and this means that the $1$-circuit is between $G$ and $F$, i.e. for all $0\leqslant i\leqslant 2n-1$ the edge $e_{i}$ is either an edge of $G$ or an edge of $F$. In this case, it is not counted as a $1$-circuit between $F\plug G$ and $H$. In the second case, the $1$-circuit goes through at least one element of $V_{H}$, and it is therefore not counted as a $1$-circuit between $G$ and $F$. In this case, the fact that this $1$-circuit induces a $1$-circuit between $F\plug G$ and $H$ is clear from the definitions. Indeed, if $i_{1},\dots,i_{k}$ are the indices such that the edges $e_{i_{j}}$ are the only edges in $\pi$ coming from $H$, then the paths defined for $0\leqslant j\leqslant k+1$ (taking $i_{0}=-1$ and $i_{k+1}=2n$) as $\pi_{j}=\{e_{p}\}_{i_{j}+1\leqslant p\leqslant i_{j+1}-1}$ are in one-to-one correspondence with edges in $F\plug G$.
\end{proof}

\begin{definition}[Measurement of alternating $1$-circuits]
Let $G$ and $H$ be two directed weighted graphs. We define, taking $log(0)=-\infty$, their interaction $\scalar{G,H}\in\mathbb{R}_{\geqslant 0}\cup\{\infty\}$ as:
\begin{equation}
\scalar{G,H}=\sum_{\pi\in\cycl{G,H}} -log(1-\omega_{G\bicol H}(\pi))\nonumber
\end{equation}
\end{definition}

\begin{remark}
The choice of the function $-log(1-x)$ is essential in order to get Proposition \ref{invariant}, which is the key result upon which the correspondence between our framework and Girard's geometry of interaction is constructed. Indeed, we will see in the last section (Theorem \ref{meashyp}) that our measurement corresponds exactly to the one used by Girard.

However, given any function $m : ]0,1]\rightarrow \mathbb{R}_{\geqslant 0}\cup\{\infty\}$, we can define $\scalar{G,H}_{m}$ as the sum $\sum_{\pi\in\cycl{G,H}} m(\omega_{G\bicol H}(\pi))$ and get all the results of sections \ref{goi} and \ref{denot}.
\end{remark}

\begin{theorem}[Adjunction]\label{adjunctweight}
Let $F$, $G$, and $H$ be directed weighted graphs such that $V_{G}\cap V_{H}=\emptyset$ and $V_{G}\cup V_{H}\subseteq V_{F}$. We have 
\begin{equation}
\scalar{F,G\cup H}=~\scalar{F,G}+\scalar{F\plug G, H}\label{adjuncteq}
\end{equation}
\end{theorem}

\begin{proof}
From Proposition \ref{adjunct}, to each alternating $1$-circuit $\pi$ in $F\bicol(G\cup H)$ there corresponds one and only one alternating $1$-circuit which lies either in $F\bicol G$ or in $(F\plug G)\bicol H$. In both cases, the definitions of weights ensure that the corresponding $1$-circuit has the same weight as $\pi$.
\end{proof}

For the next proposition, we extend the definition of $\scalar{F,G}$ to graphs with weights greater than $1$ by letting 
\begin{equation}
\scalar{F,G}=\sum_{\pi\in\cycl{F,G}} \sum_{k=1}^{\infty} \frac{(\omega_{F\bicol G}(\pi))^{k}}{k}\nonumber
\end{equation}
which allows us to consider $\scalar{\what{G},\what{H}}$.
\begin{proposition}\label{invariant}
Let $G,H$ be directed weighted graphs. We have
\begin{equation}
\scalar{G,H}=\scalar{G,\what{H}}\nonumber
\end{equation}
\end{proposition}

The proof of this proposition relies on the following technical lemma and its corollaries.

\begin{lemma}\label{once}
Let $G$ be a graph, and $e_{1}, e_{2}$ be edges with same source and target of respective weights $x_{1},x_{2}$. Let $G'$ be the graph $G$ where we replaced $e_{1}, e_{2}$ by a single edge $g$ of weight $x_{1}+x_{2}$. Let $\bar{\pi}$ be a $1$-circuit in $G'$ that goes through $g$ exactly $l$ times, i.e. $\bar{\pi}=\overline{\rho_{1}g\rho_{2}g\dots\rho_{l}g}$ where for all $1\leqslant i\leqslant l$ the path $\rho_{i}$ does not contain $g$. Let us denote by $F,E$ the following sets:
\begin{eqnarray}
F&=&\{\mu=\rho_{1}e_{i_{1,1}}\dots\rho_{l}e_{i_{1,l}}\rho_{1}e_{i_{2,1}}\dots\rho_{l}e_{i_{2,l}}\dots\rho_{1}e_{i_{m,1}}\dots\rho_{l}e_{i_{m,l}}\}\nonumber\\
E&=&\{\mu\in F~|~\mu\text{ is a $1$-cycle}\}\nonumber
\end{eqnarray}
Then $\bar{E}$ will denote the set of $1$-circuits in $E$, i.e. $\bar{E}$ is the set $E$ quotiented by cyclic permutations, and we have the following equality:
\begin{equation}
-log(1-\omega_{G'}(\bar{\pi}))=\sum_{\bar{\mu}\in\bar{E}}-log(1-\omega_{G}(\bar{\mu}))\nonumber
\end{equation}
\end{lemma}

\begin{proof}
Let us denote by $y_{i}$ the weight of the path $\rho_{i}$. Then:
\begin{eqnarray}
-log(1-\omega_{G'}(\bar{\pi})&=&\sum_{k\geqslant 1} \frac{\left((x_{1}+x_{2})^{l}y_{1}\dots y_{l}\right)^{k}}{k}\nonumber\\
&=&\sum_{k\geqslant 1} \frac{1}{k} \left(\prod_{j=1}^{l} \left((x_{1}+x_{2})y_{j}\right)^{k}\right)\nonumber\\
&=&\sum_{k\geqslant 1} \frac{1}{k} \left(\prod_{j=1}^{l} \left(\sum_{i=0}^{k}\binom{k}{i}x_{1}^{i}x_{2}^{k-i}y_{j}^{k}\right)\right)\nonumber
\end{eqnarray}

Let us denote by $F^{k}_{j}$ the set of paths $\{\rho_{j}x_{i_{1}}\dots\rho_{j}x_{i_{k}}~|~ 0\leqslant i_{p}\leqslant n\}$. Since there are exaclty $\binom{k}{i}$ elements $\mu$ of $F^{k}_{j}$ such that $\mu$ goes through $e_{1}$ exactly $i$ times, we have $\sum_{i=0}^{k}\binom{k}{i}x_{1}^{i}x_{2}^{k-i}y_{j}^{k}=\sum_{\mu\in F^{k}_{j}} \omega_{G}(\mu)$.

Moreover, the is an obvious bijection between $F^{k}=\{\mu\in F~|~|\mu|=k\}$ and the product $F^{k}_{1}\times F^{k}_{2}\times\dots\times F^{k}_{l}$. We therefore get:

\begin{eqnarray}
-log(1-\omega_{G'}(\bar{\pi}))&=&\sum_{k\geqslant 1} \frac{1}{k} \left(\prod_{j=1}^{l} \left(\sum_{\mu\in F^{k}_{j}} \omega_{G}(\mu)\right)\right)\nonumber\\
&=&\sum_{k\geqslant 1} \frac{1}{k} \left(\sum_{\mu\in F^{k}} \omega_{G}(\mu)\right)\nonumber
\end{eqnarray}

Now, if we take an element $\mu$ of $F^{k}$, it is a $d$-cycle for an integer $d$ that divides $k$ (what we will denote by $d\mathbin{|}k$). This means that there is an element $\nu\in E^{k/d}$ such that $\mu=\nu^{d}$. By Proposition \ref{countingeqclas} its equivalence class $\bar{\nu}$ up to cyclic permutations is then of cardinality $k/d$. Hence, since $\bar{E}^{k/d}$ is the set of equivalence classes up to cyclic permutations of the elements of $E^{k/d}$, we obtain: 

\begin{eqnarray}
-log(1-\omega_{G'}(\bar{\pi}))&=&\sum_{k\geqslant 1}\sum_{d\mathbin{|}k} \sum_{\nu\in E^{k/d}} \frac{(\omega_{G}(\nu))^{d}}{k}\nonumber\\
&=&\sum_{k\geqslant 1}\sum_{d\mathbin{|}k} \sum_{\bar{\nu}\in \bar{E}^{k/d}} \frac{k}{d}\frac{(\omega_{G}(\bar{\nu}))^{d}}{k}\nonumber\\
&=&\sum_{\bar{\nu}\in \bar{E}} \sum_{d\geqslant 1} \frac{(\omega_{G}(\bar{\nu}))^d}{d}\nonumber
\end{eqnarray}
\end{proof}

By a simple recurrence, we can now prove the same result for any finite number of edges $e_{1},\dots,e_{n}$.

\begin{corollary}\label{twiceandmore}
Let $G$ be a graph, and $e_{1},\dots,e_{n}$ be edges with same source and target of respective weights $x_{1},\dots,x_{n}$. Let $G'$ be the graph $G$ where we replaced $e_{1},\dots,e_{n}$ by a single edge $g$ of weight $\sum_{i=1}^{n} x_{i}$. Let $\bar{\pi}$ be a $1$-circuit in $G'$ that goes through $g$ exactly $l$ times, i.e. $\bar{\pi}=\overline{\rho_{1}g\rho_{2}g\dots\rho_{l}g}$ where for all $1\leqslant i\leqslant l$ the path $\rho_{i}$ does not contain $g$. Let us denote by $F,E$ the following sets:
\begin{eqnarray}
F&=&\{\mu=\rho_{1}e_{i_{1,1}}\dots\rho_{l}e_{i_{1,l}}\rho_{1}e_{i_{2,1}}\dots\rho_{l}e_{i_{2,l}}\dots\rho_{1}e_{i_{m,1}}\dots\rho_{l}e_{i_{m,l}}\}\nonumber\\
E&=&\{\mu\in F~|~\mu\text{ is a $1$-cycle}\}\nonumber
\end{eqnarray}
Then $\bar{E}$ will denote the set of $1$-circuits in $E$, i.e. $\bar{E}$ is the set $E$ quotiented by cyclic permutations, and we have the following equality:
\begin{equation}
-log(1-\omega_{G'}(\bar{\pi}))=\sum_{\bar{\mu}\in\bar{E}}-log(1-\omega_{G}(\bar{\mu}))\nonumber
\end{equation}
\end{corollary}

But the result is actually true even when one has an infinite (countable) number of edges.

\begin{corollary}\label{infinitetimes}
Let $G$ be a graph, and $(e_{i})_{i\in\mathbb{N}}$ be edges with the same sources and targets. For all $i\in\mathbb{N}$, we denote by $x_{i}$ the weight of $e_{i}$ in $G$. Let $G'$ be the graph $G$ where we replaced $e_{0},\dots,e_{n},\dots$ by a single edge $g$ of weight $\sum_{i\in\mathbb{N}}x_{i}$. Let us consider a $1$-circuit $\bar{\pi}$ in $G'$ that goes through $g$ exactly $l$ times, i.e. $\bar{\pi}=\overline{\rho_{1}g\rho_{2}g\dots\rho_{l}g}$ where for all $1\leqslant i\leqslant l$ the path $\rho_{i}$ does not contain $g$. Let us denote by $F,E$ the following sets:
\begin{eqnarray}
F&=&\{\mu=\rho_{1}e_{i_{1,1}}\dots\rho_{l}e_{i_{1,l}}\rho_{1}e_{i_{2,1}}\dots\rho_{l}e_{i_{2,l}}\dots\rho_{1}e_{i_{m,1}}\dots\rho_{l}e_{i_{m,l}}\}\nonumber\\
E&=&\{\mu\in F~|~\mu\text{ is a $1$-cycle}\}\nonumber
\end{eqnarray}
Then $\bar{E}$ will denote the set of $1$-circuits in $E$, i.e. $\bar{E}$ is the set $E$ quotiented by cyclic permutations, and we have the following equality:
\begin{equation}
-log(1-\omega_{G'}(\bar{\pi}))=\sum_{\bar{\mu}\in\bar{E}}-log(1-\omega_{G}(\bar{\mu}))\nonumber
\end{equation}
\end{corollary}

\begin{proof}
Let us first introduce some notations. We will consider the sets $F_{\leqslant i}$ defined, for every $i\in\mathbb{N}$, as the set of cycles $\pi$ in $F$ such that $e_{k}\in\pi\Rightarrow k\leqslant i$. This allows us to define $F_{i}=F_{\leqslant i}-F_{\leqslant i-1}$ for $i\leqslant 1$ and $F_{0}=F_{\leqslant 0}$ by convention. Notice that $(F_{i})_{i\in\mathbb{N}}$ is a partition of $F$.  Following the preceding notations, we will denote by $E_{i}$ (resp. $E_{\leqslant i}$) the set of $1$-cycles in $F_{i}$ (resp. $F_{\leqslant i}$) and by $\bar{E}_{i}$ (resp. $\bar{E}_{\leqslant i}$) the corresponding set of $1$-circuits.

Then, by continuity of the logarithm and the preceding corollary, we have:
\begin{eqnarray}
-log(1-\omega_{G'}(\bar{\pi}))&=&\lim_{n\to \infty} -log(1-(\sum_{i=0}^{n}x_{i})y)\nonumber\\
&=&\lim_{n\to \infty} \sum_{\bar{\mu}\in\bar{E}_{\leqslant n}} -log(1-\omega_{G}(\bar{\mu}))\nonumber\\
&=&\lim_{n\to\infty} \sum_{i=0}^{n}\sum_{\bar{\mu}\in\bar{E}_{i}} -log(1-\omega_{G}(\bar{\mu}))\nonumber\\
&=&\sum_{n=0}^{\infty} \sum_{\bar{\mu}\in\bar{E}_{i}} -log(1-\omega_{G}(\bar{\mu}))\nonumber\\
&=&\sum_{\bar{\mu}\in\bar{E}} -log(1-\omega_{G}(\bar{\mu}))\nonumber
\end{eqnarray}
\end{proof}

\begin{proof}[Proof of Proposition \ref{invariant}]~\\
Using the preceding Lemma (\ref{once}) and its corollaries (\ref{twiceandmore}, \ref{infinitetimes}), we obtain that contracting all the edges with same source and target does not change the measurement.

Then, we obtain the wanted general result stated in Proposition \ref{invariant} by an iteration of this result over every set of vertices $E_{H}(v,w)$ ($v,w\in V_{H}$) in the graph $H$.
\end{proof}

\section{Geometry of Interaction}\label{goi}

Now we will construct a geometry of interaction based on the three-terms adjunction we obtained. Our objects (projects) will consist of a weighted graph, obviously, but we add to this a real number. This real number is here to compensate for the additional term $\scalar{F,G}$ of the adjunction (see the remark following Theorem \ref{duality}).

\begin{definition}[Projects]
A \emph{project} is a couple $\de{a}=(a,A)$, where $a\in \mathbb{R}_{\geqslant 0}$ is called the \emph{wager}, and $A$ is a weighted directed graph over a finite set of vertices $V_{A}$. The set $V_{A}$ of vertices of $A$ will be called the \emph{carrier} of $\de{a}$. 
\end{definition}

\begin{definition}[Measurement of the interaction]
Let $\de{a}=(a,A)$ and $\de{b}=(b,B)$ be two projects. We define $\sca{a}{b}=a+b+\scalar{A,B}$.
\end{definition}

\begin{definition}[Orthogonality]
Two projects $\de{a}$ and $\de{b}$ of same carrier are said to be \emph{orthogonal} when $\sca{a}{b}\not\in\{0,\infty\}$. We denote it by $\de{a}\simperp\de{b}$ and we define the orthogonal set of a set $A$ of projects of same carrier as $A^{\pol}=\{\de{b}~|~\forall \de{a}\in A,~ \de{a}\simperp\de{b}\}$.
\end{definition}

\begin{remark}\label{proofnet}
We want here to stress an important point related to proof nets. Taking a proof net $\finhyp$, the switchings used in the long-trip criterion define a set of permutations over the atoms: the permutation $\sigma$ induced by the axiom links, and permutations $\tau_{S}$ induced by the remaining links (one for each switching $S$). The correctness criterion tells\footnote{This is an easy reformulation of the Long Trips criterion that can be found in Girard's courses \cite{blindspot}.} us that the proof net is correct if and only if for all switching $S$ the product of the permutations $\sigma$ and $\tau_{S}$ is cyclic. Now, the permutations $\sigma$ and $\tau$ define two graphs $S$ and $T$ such that there exists exactly one alternating cycle in $S\bicol T$, going through all links. If we modify the weight of one of the edges of $T$ to make it strictly less than $1$, we then obtain a graph $T'$ such that $(0,T')\simperp(0,S)$.

We can therefore see switchings as projects that are orthogonal to the axiom links. Conversely, we can consider projects as generalized switching induced permutations.
\end{remark}

Now that the objects and the duality between them have been defined, we can introduce conducts — that will correspond to formulas or types — as sets of objects equal to their biorthogonal.
\begin{definition}[Conducts]
A non-empty set of projects $\cond{S}$ of same carrier $X$ equal to its biorthogonal $\cond{S}^{\pol\pol}$ is called a \emph{conduct}. We will call $X$ the carrier of the conduct $\cond{S}$.
\end{definition}

\begin{remark}
As for any definition of orthogonality, we get, for any sets $A,B$ of designs (of the same carrier), the classical results:
\begin{itemize}
\item $A^{\pol\pol\pol}=A^{\pol}$;
\item $A\subseteq B\Rightarrow B^{\pol}\subseteq A^{\pol}$.
\end{itemize}
\end{remark}

We will now proceed to define connectives on projects, and then on conducts.

\begin{definition}[Tensor]
The \emph{tensor product} of projects of disjoint carriers is defined as:
\begin{equation}
(a,A)\otimes(b,B)=(\sca{a}{b},A\cup B)\nonumber
\end{equation}
\end{definition}

\begin{remark} 
Notice that in this definition, since $\de{a}$ and $\de{b}$ have disjoint carriers, $\sca{a}{b}=a+b$.
\end{remark}

\begin{definition}[Tensor on Conducts]
Let $\cond{A,B}$ be conducts of disjoint carrier. We can form the conduct $\cond{A}\otimes\cond{B}$
\begin{equation}
\cond{A\otimes B}=\{\de{a}\otimes\de{b}~|~\de{a}\in\cond{A},\de{b}\in\cond{B}\}^{\pol\pol}\nonumber
\end{equation}
We will denote by $\cond{A\odot B}$ the set $\{\de{a}\otimes\de{b}~|~\de{a}\in\cond{A},\de{b}\in\cond{B}\}$.
\end{definition}

\begin{definition}[Cut]
We define, when $\sca{f}{g}\not= \infty$, the \emph{cut} $\de{f\plug g}$ of the projects $\de{f}$ and $\de{g}$ as follows:
\begin{equation}
\de{f}\plug\de{g}=(\sca{f}{g},F\plug G)\nonumber
\end{equation}
\end{definition}

\begin{proposition}[Properties of the Tensor]
The tensor product is commutative and associative. Moreover it has a neutral element\footnote{Here our notation differs from Girard's \cite{goi5}, where the unit conduct of the tensor is denoted by $\top$, which in his framework is also the unit conduct of the additive conjunction.}, namely $1=\{(0,(\emptyset,\emptyset))\}^{\pol\pol}=\{(a,(\emptyset,\emptyset))~|~a\geqslant 0\}$, where $(\emptyset,\emptyset)$ denotes the empty graph.
\end{proposition}

\begin{definition}[Linear Implication]
Let $\cond{A,B}$ be conducts of disjoint carriers.
\begin{equation}
\cond{A\multimap B}=\{\de{f}~|~\forall \de{a}\in \cond{A}, \de{f}\plug\de{a}\downarrow\in\cond{B}\}
\end{equation}
where the arrow means that $\de{f}\plug\de{a}$ is defined.
\end{definition}

The fact that this defines a conduct is justified by the following proposition.

\begin{theorem}[Duality]\label{duality}
We have that:
\begin{equation}
\cond{A\multimap B}=(\cond{A}\otimes \cond{B}^{\pol})^{\pol}\nonumber
\end{equation}
\end{theorem}

\begin{proof}
Let $V_{A}$ and $V_{B}$ be the disjoint carriers of $\cond{A}$ and $\cond{B}$, let $\de{f}$ be a project of carrier $V_{A}\cup V_{B}$, and let $\de{a,b}$ be projects in $\cond{A}$ and $\cond{B}^{\pol}$ respectively. From the adjunction (Theorem \ref{adjunctweight}, see also the following remark) we have the equivalence between $\de{f}\simperp\de{a}\otimes\de{b}$ and $\de{f}\plug\de{a}\simperp\de{b}$. We thus get that $\de{f}\in\cond{A\multimap B}$ if and only if $\de{a}\otimes\de{b}\simperp\de{f}$ — which means that $\de{f}\in\cond{(A\odot B^{\pol})}^{\pol}=\cond{(A\otimes B^{\pol})}^{\pol}$.
\end{proof}

\begin{remark}
The adjunction implies that $\de{f}\simperp\de{a}\otimes\de{b}$ is equivalent to $\de{f}\plug\de{a}\simperp\de{b}$, but it moreover tells us the interaction is exactly the same. Indeed, if either $\de{f}\simperp\de{a}\otimes\de{b}$ or $\de{f}\plug\de{a}\simperp\de{b}$, we have:
\begin{eqnarray*}
\sca{f}{a\otimes b}&=&a+b+f+\scalar{F,A\cup B}\\
&=&a+b+f+\scalar{F,A}+\scalar{F\plug A,B}\\
&=&\sca{f \plug a}{b}
\end{eqnarray*}
By the way, we can see in this computation how the wager compensates for the additional term in the adjunction. Indeed, the wager can be seen as a residue $\scalar{F,A}$ of the composition of graphs, a residue of the internal cycles (as in Figure \ref{internal}) that may appear when plugging $F$ and $A$.
\end{remark}

\begin{proposition}[Mix Rule]\label{mixrule}
Let $\de{a}\poll\de{b}$ and $\de{c}\poll\de{d}$ be projects such that the carrier of $\de{a}$ (and therefore of $\de{b}$) is disjoint from the carrier of $\de{c}$ and $\de{d}$. Then $\de{a}\otimes\de{c}\poll\de{b}\otimes\de{d}$. As a consequence, we have $\cond{A\otimes B}\subset \cond{A\parr B}=\cond{A}^{\pol}\multimap\cond{B}$ for any conducts $\cond{A,B}$ of disjoint carriers.
\end{proposition}

\begin{proof}
It is immediate that $\sca{a\otimes c}{b\otimes d}=\sca{a}{b}+\sca{c}{d}$. Since both summands are non-zero positive reals, their sum is a non-zero positive real, hence $\de{a\otimes c}\poll\de{b\otimes d}$. Now, let $\de{a}$ and $\de{b}$ be two projects in conducts $\cond{A}$ and $\cond{B}$ of disjoint carriers, we just showed that $\de{a}\otimes\de{b}\in (\cond{A}^{\pol}\otimes\cond{B}^{\pol})^{\pol}$. By Proposition \ref{duality}, we have that $(\cond{A}^{\pol}\otimes\cond{B}^{\pol})^{\pol}=\cond{A}^{\pol}\multimap\cond{B}$, hence $\cond{A\otimes B}\subset\cond{A^{\pol}\multimap B}$.
\end{proof}

Eventually, we define an important object that will be used in the next section.
\begin{definition}[Delocations]\label{deloc}
Let $\de{a}$ be a project of carrier $V_{A}$, $V_{B}$ a set such that $V_{A}\cap V_{B}=\emptyset$, and $\phi: V_{A} \rightarrow V_{B}$ a bijection. We define the \emph{delocation} of $\de{a}=(a,A)$ as $\phi(\de{a})=(a,\phi(A))$, where $\phi(A)$ is exactly the same graph as $A$ on the set of vertices $V_{B}$.
\end{definition}

\begin{remark}
For the sake of simplicity, we will use abusively — mainly in the next section — the notation $\phi(\de{a})$ even when the bijection $\phi$ does not satisfy $dom(\phi)\cap codom(\phi)=\emptyset$ (hence we do not necessarily have $V_{A}\cap\phi(V_{A})=\emptyset$). However, this amounts to nothing more than a simplification: if we define 
\begin{equation}
\begin{array}{lrclcrcl}
\phi': & dom(\phi)\times\{0\} & \rightarrow & codom(\phi)\times\{1\},&~~~ &(x,0) & \mapsto & (\phi(x),1)\\
\iota: & dom(\phi)  & \rightarrow & dom(\phi)\times\{0\},&& x & \mapsto & (x,0)\\
\zeta: & codom(\phi)\times\{1\} & \rightarrow & codom(\phi),&& (x,1) & \mapsto & x\end{array}\nonumber
\end{equation}
 then what we abusively denote by $\phi(\de{a})$ is correctly defined through delocations by $\zeta(\phi'(\iota(\de{a})))$.
\end{remark}

\begin{proposition}
Keeping the notations of Definition \ref{deloc}, we define the project $\de{Fax}_{\phi}=(0,\Phi)$ with
\begin{eqnarray*}
E_{\Phi}&=&\{(a,\phi(a))~|~ a\in V_{A}\}\cup\{(\phi(a),a)~|~ a \in V_{A}\}\\
\Phi&=&(V_{A}\cup V_{B},E_{\Phi},\omega_{\Phi}(e)=1)
\end{eqnarray*}
Then $\de{Fax}_{\phi}\in\cond{A}\multimap\phi(\cond{A})$.
\end{proposition}

\section{Denotational Semantics}\label{denot}

We will now prove that our geometry of interaction yields a denotational model of MLL by showing that we can define a $\ast$-autonomous category from it. This $\ast$-autonomous category has the interesting peculiarity of interpreting the multiplicative units by different objects, contrarily to many of the known categorical models of MLL, such as the relational model or the coherence spaces of Girard. Most of this section consists in proving that our category has the required properties, but these technicalities hide the principal interest of explicitly defining the category. We want to stress here the differences between geometry of interaction and denotational semantics. In particular, even if the objects of the category still have a location, we are not working in a \emph{locative} framework since all our definitions (morphisms, composition, functors) are given on delocations of the objects.

The difference comes from the fact that geometry of interaction is a semantics of processes, of actions. This is why there are no elements in $\cond{A\multimap A}$: a process that for all $a$ yields $a$ is not a process, performs no action. The objects that are closest to the identity are the delocations, i.e. a function that makes a copy of $a$ in another location. When defining the category we have to consider delocations as identity maps in order to have some identity morphisms. It is this “quotient" by delocations that implies the loss of locativity.

Before defining the category, we define two functions $\mathbb{N}\rightarrow \mathbb{N}\times\{0,1\}$
\begin{equation}
\left.\begin{array}{lrcl} 
\psi_{0}: & x & \mapsto & (x,0)\\
\psi_{1}: & x & \mapsto & (x,1)
\end{array}\right.
\nonumber
\end{equation}

\begin{definition}[Objects and morphisms of \catmll]
We define the following category:
\begin{equation}
\left.\begin{array}{l}\mathfrak{Obj}=\{\cond{A}~|~\cond{A}=\cond{A}^{\pol\pol}\text{ of carrier }X_{\cond{A}}\subset\mathbb{N}\}\\
\mathfrak{Mor}[\cond{A},\cond{B}]=\{\de{f}\in \psi_{0}(\cond{A})\multimap\psi_{1}(\cond{B})\}\end{array}\right.\nonumber
\end{equation}
\end{definition}

To define the composition of morphisms, we will use in fact three copies of $\mathbb{N}$. We thus define the following useful bijections:
\begin{center}
$\left\{\begin{array}{lrcl}\mu: \!\!\!& \mathbb{N}\times\{0,1\} & \rightarrow & \mathbb{N}\times\{1,2\}\\
                                                 &  (x,i) & \mapsto & (x,i+1)\end{array}\right.$
\hfill
$\left\{\begin{array}{lrcl}\nu: \!\!\!& \mathbb{N}\times\{0,2\} & \rightarrow & \mathbb{N}\times\{0,1\}\\
                                                 &  (x,0) & \mapsto & (x,0)\\
                                                 &  (x,2) & \mapsto & (x,1)\end{array}\right.$
\end{center}

\begin{definition}[Composition in \catmll]
Given two morphisms $\de{f}$ and $\de{g}$ in $\mathfrak{Mor}[\cond{A},\cond{B}]$ and $\mathfrak{Mor}[\cond{B},\cond{C}]$ respectively, we define 
\begin{equation}
\de{g}\circ\de{f}=\nu(\de{f}\plug\mu(\de{g}))\nonumber
\end{equation}
\end{definition}

\begin{proposition}[\catmll is a Category]\label{category}
The sets of objects and morphisms we just defined, together with the composition induced by the reduction of graphs, is a category.
\end{proposition}

\begin{proof}~\\ We first show that there exists an identity morphism for every object in \catmll, and that it is the neutral for the composition.
\begin{itemize}
\item \textbf{Unit} For $S\subset\mathbb{N}$, define the bijection 
\begin{equation}
1_{S}: \left\{\begin{array}{rcl}S\times\{0\}&\rightarrow& S\times\{1\}\\ (x,0)&\mapsto&(x,1)\end{array}\right.\nonumber
\end{equation}
Then $\de{Fax}_{1_{S}}$ is the identity morphism for all objects of carrier $S\subset\mathbb{N}$.

A simple computation shows that the required diagram commutes.

\item \textbf{Associativity} The fact that the composition is associative follows directly from Proposition \ref{assocmll}.
\end{itemize}
\end{proof}

It is well-known that a $*$-autonomous category yields a model of MLL \cite{seely89}. We shall now build a $*$-autonomous structure on \catmll. We begin by defining a monoidal functor $\bar{\otimes}$ and show that we have a symmetric monoidal closed category. Then we will show that the object $\bot=1^{\pol}$ is dualizing, meaning the category is $*$-autonomous.

\begin{definition}
A monoidal category is a category $\mathbb{K}$ with a bifunctor $\otimes: \mathbb{K}\times\mathbb{K}\rightarrow\mathbb{K}$, a (left and right) unit $1\in \mathfrak{Obj}_{\mathbb{K}}$, satisfying $(A\otimes B)\otimes C\cong A\otimes (B\otimes C)$. In addition, some diagrams concerning associativity and the unit must commute (we refer to Mac Lane \cite{mclane} for a complete definition).\\
It is said to be symmetric when we have $A\otimes B\cong B\otimes A$, and closed when we can associate to each set of morphisms $\mathfrak{Mor}_{\mathbb{K}}[A,B]$ an object $A\rightarrow B\in\mathfrak{Obj}_{\mathbb{K}}$ such that $\mathfrak{Mor}_{\mathbb{K}}[A\otimes B,C]$ is naturally isomorphic to $\mathfrak{Mor}_{\mathbb{K}}[A,B\rightarrow C]$.
\end{definition}

In order to define the bifunctor, we will use the functions $\phi: \mathbb{N}\times\{0,1\}\rightarrow \mathbb{N}$ defined by $\phi((x,i))=2x+i$ and $\tau$
\begin{equation}
\tau: \left\{\begin{array}{rcl}
\mathbb{N}\times\{0,1\}&\rightarrow&\mathbb{N}\times\{0,1\}\\
(2x+1,0)&\mapsto& (2x,1)\\
(2x,1)&\mapsto&(2x+1,0)\\
(x,i)&\mapsto&(x,i)\text{ otherwise}\end{array}\right.\nonumber
\end{equation}
\begin{proposition}
The category $($\catmll$,\bar{\otimes},1)$ is a symmetric monoidal closed category, where the bifunctor $\bar{\otimes}$ is induced by the tensor product defined on objects by
\begin{equation}
\cond{A}\bar{\otimes}\cond{B}=\phi(\psi_{0}(\cond{A})\otimes\psi_{1}(\cond{B}))\nonumber
\end{equation}
and on morphisms as
\begin{equation}
\de{f}\bar{\otimes} \de{g}=\tau(\psi_{0}(\phi(\de{f}))\otimes\psi_{1}(\phi(\de{g})))\nonumber
\end{equation}
and the unit is the conduct $1=\{(0,(\emptyset,\emptyset))\}^{\pol\pol}$,
\end{proposition}

\begin{proof} We have to check first that it is a monoidal category, and then that it is symmetric. We will define the isomorphisms by bijections from $\mathbb{N}$ onto $\mathbb{N}$. Indeed, such a bijection $\alpha$ induces an isomorphism for any $S\subset\mathbb{N}$ by letting
\begin{equation}
\overline{\alpha}_{S}=(0,A_{S})\nonumber
\end{equation}
with $A_{S}$ the weighted graph
\begin{eqnarray*}
V_{A_{S}}&=&(S\times\{0\})\cup(\alpha(S)\times\{1\})\\
E_{A_{S}}&=&\{((x,0),(\alpha(x),1))\}\cup\{((\alpha(x),1),(x,0))\}
\end{eqnarray*}
where all edges are of weight $1$.

\begin{itemize}
\item \textbf{Associativity} Let $\cond{A},\cond{B},\cond{C}$ be three objects of \catmll. For any conducts on disjoint carriers, and any delocation $\theta$, we have $\theta(\cond{A}\otimes\cond{B})=\theta(\cond{A})\otimes\theta(\cond{B})$ because the conducts have a disjoint carrier. We can therefore see $\cond{A}\bar{\otimes}(\cond{B}\bar{\otimes} \cond{C})$ and $(\cond{A}\bar{\otimes}\cond{B})\bar{\otimes} \cond{C}$ as the (localized) tensor product of delocations of $\cond{A,B,C}$, i.e.
\begin{eqnarray}
\cond{A}\bar{\otimes}(\cond{B}\bar{\otimes}\cond{C})&=&\phi(\psi_{0}(\cond{A})\otimes\psi_{1}(\phi(\psi_{0}(\cond{B})\otimes\psi_{1}(\cond{C})))))\nonumber\\
&=&\phi(\psi_{0}(\cond{A}))\otimes\phi(\psi_{1}(\phi(\psi_{0}(\cond{B}))))\otimes\phi(\psi_{1}(\phi(\psi_{1}(\cond{C}))))\nonumber\\
(\cond{A}\bar{\otimes}\cond{B})\bar{\otimes}\cond{C}&=&\phi(\psi_{0}(\phi(\psi_{0}(\cond{A})\otimes\psi_{1}(\cond{B})))\otimes\psi_{1}(\cond{C}))\nonumber\\
&=&\phi(\psi_{0}(\phi(\psi_{0}(\cond{A}))))\otimes\phi(\psi_{0}(\phi(\psi_{1}(\cond{B}))))\otimes\phi(\psi_{1}(\cond{C}))\nonumber
\end{eqnarray}

Once we noticed this, we are left with a simple combinatorics problem, and we easily verify that the following bijection, which does not depend on the objects considered, transforms $\cond{A}\bar{\otimes}(\cond{B}\bar{\otimes}\cond{C})$ into $(\cond{A}\bar{\otimes}\cond{B})\bar{\otimes}\cond{C}$
\begin{equation}
\alpha: n \mapsto \left\{\begin{array}{ll}
2n&\text{ if }n\equiv 0[2]\\
n+1&\text{ if }n\equiv  1[4]\\
(n-1)/2&\text{ if }n\equiv 3[4]
\end{array}\right.
\nonumber
\end{equation}
Hence, we get the associativity up to a natural transformation. Moreover, it satisfies the required pentagonal diagram.
\item \textbf{Unit} The unit satisfies that there exists two natural transformations $\lambda: 1 \bar{\otimes}\cond{A}\cong \cond{A}$ and $\rho: \cond{A}\bar{\otimes}1\cong \cond{A}$. Indeed, we only have to define:
\begin{equation}
\lambda=\rho=\pi\circ\phi^{-1}\nonumber
\end{equation}
where $\pi: \mathbb{N}\times\{0,1\}\rightarrow \mathbb{N}$ are defined as $\pi(n,i)=n$.
Since the required diagram commutes, we have that \catmll is a monoidal category.
\item \textbf{Closure} We already saw that $\cond{A\multimap B}$ is a conduct in the preceding section. Moreover, if $X$ and $Y$ are the carriers of $\cond{A}$ and $\cond{B}$, the conduct $\cond{\phi(\psi_{0}(A)\multimap\psi_{1}(B))}$ is of carrier $\phi(\psi_{0}(X)\cup\psi_{1}(Y))\subset\mathbb{N}$, hence an object of \catmll. Denoting it by $\cond{A}\bar{\multimap}\cond{B}$, we have $\mathfrak{Mor}[\cond{A}\bar{\otimes}\cond{B},\cond{C}]\cong\mathfrak{Mor}[\cond{A},\cond{B}\bar{\multimap}\cond{C}]$ from Theorem \ref{duality} and the associativity isomorphism $\alpha$:
\begin{eqnarray*}
\mathfrak{Mor}[\cond{A}\bar{\otimes} \cond{B}, \cond{C}]&=&(\psi_{0}(\phi(\psi_{0}(\cond{A})\otimes\psi_{1}(\cond{B})))\otimes\psi_{1}(\cond{C})^{\pol})^{\pol}\\
&\stackrel{\phi^{-1}\alpha^{-1}\phi}{\cong}&(\psi_{0}(\cond{A})\otimes\psi_{1}(\phi(\psi_{0}(\cond{B})\otimes\psi_{1}(\cond{C})^{\pol})^{\pol\pol}))^{\pol}\\
&=&(\psi_{0}(\cond{A})\otimes\psi_{1}(\cond{B}\bar{\multimap}\cond{C})^{\pol})^{\pol}\\
&=&\psi_{0}(\cond{A})\multimap\psi_{1}(\cond{B}\bar{\multimap}\cond{C})\\
&=&\mathfrak{Mor}[\cond{A},\cond{B}\bar{\multimap}\cond{C}]
\end{eqnarray*}
\item \textbf{Symmetry} The following bijection can be defined:
\begin{equation}
\gamma: \left\{\begin{array}{rcl}\mathbb{N}&\rightarrow &\mathbb{N}\\
2n&\mapsto& 2n+1\\
2n+1&\mapsto& 2n\end{array}\right.\nonumber
\end{equation}
This bijection defines the isomorphism between $\cond{A}\bar{\otimes}\cond{B}$ and $\cond{B}\bar{\otimes}\cond{A}$. This isomorphism is natural, and since $\gamma^{2}=Id$ we obtain the commutativity of the diagram for the symmetry. We eventually verify by a straightforward computation that (one of) the hexagonal braiding diagrams commute.
\end{itemize}
\end{proof}

\begin{definition}
A $*$-autonomous category $\mathbb{K}$ is a symmetric monoidal closed category $(\mathbb{K},\otimes,1)$ together with a dualizing object $\bot$.
\end{definition}

\begin{proposition}
The object $\bot=1^{\pol}$ is dualizing for \catmll.
\end{proposition}

\begin{proof} Taking the identity morphism from $\cond{A}\bar{\multimap} \bot$ to itself, we get a morphism from $(\cond{A}\bar{\multimap}\bot)\bar{\otimes} \cond{A}$ to $\bot$ by applying\footnote{See the proof of the closure of the category.} $\phi^{-1}\alpha^{-1}\phi$. From this, we get a morphism from $\cond{A}\bar{\otimes}(\cond{A}\bar{\multimap}\bot)$ to $\bot$ by precomposing with $\gamma$. Hence, applying $\phi^{-1}\alpha\phi$, we get a morphism from $\cond{A}$ to $(\cond{A}\bar{\multimap}\bot)\bar{\multimap}\bot$ defined by the function $x\mapsto 4x$. It is then an isomorphism, which means that $\bot$ is indeed dualizing.
\end{proof}

As a consequence, we get the following theorem.
\begin{theorem}
The category \catmll is a $*$-autonomous category.
\end{theorem}

\begin{remark} The tensor unit $\cond{1}$ and its dual $\bot$ are not interpreted as the same objects, contrarily to other categorical semantics of multiplicative linear logic with units. Indeed, $\cond{1}$ contains the project $(0,0)$, where $0$ denotes the empty graph on an empty set of vertices, whereas $\bot$ does not.
\end{remark}

\section{Truth}\label{truth}

We can also define a notion of truth inside our framework. We first define a successful project — what corresponds to a correct proof — to be a graph that looks like a set of axiom links, i.e. which is a disjoint union of transpositions. The idea is that a set of axiom-links that interacts correctly with the set of tests (switchings in the case of proof structures) defines a successful proof (a correct structure, i.e. a proof net). Before defining the notion of success, we need to introduce some notations.

First, we will denote by $A^{k}$ the graph of paths of length $k$ in the graph $A$.

Moreover, we define the trace of graph (a mere generalization of the trace of a matrix):
\begin{equation}
Tr(A)= \sum_{v\in V_{A}}\sum_{e\in E_{A}(v,v)} \omega_{A}(e)\nonumber
\end{equation}

We will also say a graph $G$ is \emph{symmetric} when for all vertices $v,w$ there is a weight-preserving bijection between $E_{G}(v,w)$ and $E_{G}(w,v)$.

\begin{definition}[Successful projects]
A project $\de{a}=(a,A)$ is \emph{successful} when $a=0$, the graph $\what{A}$ is symmetric, and is such that $(\what{A})^{~\!3}=\what{A}$ and $Tr(A)=0$.
\end{definition}

\begin{remark}
This definition of truth can be weakened by forgetting about the condition $Tr(A)=0$. All remaining propositions and theorems of this section would still be true (if one replaces “disjoint union of transpositions" by “disjoint union of transpositions and fix points"). We chose to present this particular definition because it better corresponds to our intuition of successful projects as a set of axiom links.

The restriction to graphs such that $Tr(A)=0$ seems moreover necessary to obtain a completeness result.
\end{remark}

\begin{proposition}\label{characterization}
If $\de{a}=(0,A)$ is successful, the graph $\what{A}$ is a disjoint union of transpositions.
\end{proposition}

\begin{proof}
The fact that $\what{A}$ is symmetric and satisfies $\what{A}^{3}=\what{A}$ implies that a given vertex cannot be the target of more than one edge, or the source of more than one edge. Indeed, let $e=(v,w)$ and $f=(w,z)$ be two edges in $\what{A}$. Then, there exist edges $e^{-1}$ and $f^{-1}$ from respectively $w$ to $v$ and from $z$ to $w$. Then there is in $\what{A}^3$ more than one edge between $v$ and $w$, namely $ee^{-1}e$ and $eff^{-1}$. Hence it cannot be equal to $\what{A}$.

We have that $\what{A}^3=\what{A}$ implies that all weights equal $1$ (since all weights $\lambda$ satisfy $\lambda^3=\lambda$), which means that $\what{A}$ is just the graph induced by a disjoint union of transpositions and fix points. However, $\what{A}$ cannot contain any fix points, since $Tr(A)=Tr(\what{A})=0$. Thus $\what{A}$ is a disjoint union of transpositions.
\end{proof}

\begin{definition}[Truth]
A conduct of carrier $V$ is \emph{true} if it contains a successful project.
\end{definition}

\begin{theorem}[Consistency]
The conducts $\cond{A}$ and $\cond{A}^{\pol}$ cannot both be true.
\end{theorem}

\begin{proof}
Suppose there exists two successful projects $\de{a}=(0,A)$ and $\de{b}=(0,B)$ in $\cond{A}$ and $\cond{A}^{\pol}$ respectively. All weights in $\what{A}$ and $\what{B}$ being equal to $1$ from Proposition \ref{characterization}, if the graph $\what{A}\bicol \what{B}$ contains no $1$-circuits, we get $\scalar{\what{A},\what{B}}=\scalar{A,B}=\sca{a}{b}=0$, and if it contains at least one $1$-circuit we have $\scalar{\what{A},\what{B}}=\scalar{A,B}=\sca{a}{b}=\infty$. Since both cases contradict the fact that $\de{a}\simperp\de{b}$, we are done.
\end{proof}

To prove compositionality, we will use the following lemma.
\begin{lemma}\label{appli}
Let $A$ and $B$ be two graphs. Then
\begin{equation}
\what{\what{A}\plug\what{B}}=\what{A\plug B}\nonumber
\end{equation}
\end{lemma}

Before going through the proof, we show on a simple example how the argument works. Taking the two graphs $A$ and $B$ of Figure \ref{ABlemmaappli}, the graphs $A\bicol B$ and $\what{A}\bicol\what{B}$ are given in Figure \ref{pluglemmaappli}. The graphs $\what{A\plug B}$ and $\what{\what{A}\plug\what{B}}$ are both composed of one edge from $a$ to $c$, and their weights are respectively equal to $x_{1}y_{1}+x_{1}y_{2}+x_{2}y_{1}+x_{2}y_{2}$ and $(x_{1}+x_{2})(y_{1}+y_{2})$, hence equal. In fact, the proof relies solely on the distribution of the multiplication over the addition.

\begin{figure}[h]
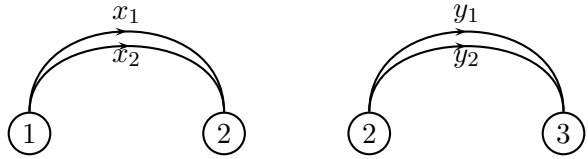

\begin{center}
$\begin{array}{ccc}
~\\
\psmatrix[colsep=1cm,rowsep=1cm]
&~&\\
[mnode=circle]1&&[mnode=circle]2
\psset{ArrowInside=->}
\nccurve[angleA=90,angleB=90,ncurvA=0.9,ncurvB=0.9]{2,1}{2,3}
\nccurve[angleA=90,angleB=90,ncurvA=1.1,ncurvB=1.1]{2,1}{2,3}
\nput{90}{1,2}{x_{1}}
\nput{-90}{1,2}{x_{2}}
\endpsmatrix
&~~~~~&
\psmatrix[colsep=1cm,rowsep=1cm]
&~&\\
[mnode=circle]2&&[mnode=circle]3
\psset{ArrowInside=->}
\nccurve[angleA=90,angleB=90,ncurvA=0.9,ncurvB=0.9]{2,1}{2,3}
\nccurve[angleA=90,angleB=90,ncurvA=1.1,ncurvB=1.1]{2,1}{2,3}
\nput{90}{1,2}{y_{1}}
\nput{-90}{1,2}{y_{2}}
\endpsmatrix
\end{array}$
\end{center}
\caption{The graphs $A$ and $B$}\label{ABlemmaappli}
\end{figure}

\begin{figure}[h]
\begin{center}
$\begin{array}{ccc}
\psmatrix[colsep=0.8cm,rowsep=0.8cm]
&~&&~&\\
[mnode=circle]1&&[mnode=circle]2&&[mnode=circle]3\\
&~&&~&
\psset{ArrowInside=->}
\nccurve[angleA=90,angleB=90,ncurvA=0.9,ncurvB=0.9]{2,1}{2,3}
\nccurve[angleA=90,angleB=90,ncurvA=1.1,ncurvB=1.1]{2,1}{2,3}
\nput{90}{1,2}{x_{1}}
\nput{-90}{1,2}{x_{2}}
\nccurve[angleA=-90,angleB=-90,ncurvA=1,ncurvB=1]{2,3}{2,5}
\nccurve[angleA=-90,angleB=-90,ncurvA=1.2,ncurvB=1.2]{2,3}{2,5}
\nput{90}{3,4}{y_{1}}
\nput{-90}{3,4}{y_{2}}
\endpsmatrix
&
~~~~~
&
\psmatrix[colsep=0.8cm,rowsep=0.8cm]
&~&&~&\\
[mnode=circle]1&&[mnode=circle]2&&[mnode=circle]3\\
&~&&~&
\psset{ArrowInside=->}
\nccurve[angleA=90,angleB=90,ncurvA=1,ncurvB=1]{2,1}{2,3}
\nput{90}{1,2}{x_{1}+x_{2}}
\nccurve[angleA=-90,angleB=-90,ncurvA=1,ncurvB=1]{2,3}{2,5}
\nput{-90}{3,4}{y_{1}+y_{2}}
\endpsmatrix
\end{array}$
\end{center}
\caption{The graphs $A\bicol B$ (on the left) and $\what{A}\bicol\what{B}$ (on the right)}\label{pluglemmaappli}
\end{figure}

\begin{proof}
An edge $f_{0}$ in $\what{\what{A}\plug \what{B}}$ is an alternating path $\pi=\epsilon_{1}\dots\epsilon_{k}$, where the $\epsilon_{i}$ are either in $\what{A}$ or in $\what{B}$ according to the parity of $i$. Write $s_{i}$ and $t_{i}$ the source and targets of the edge $\epsilon_{i}$ (for $i=1,\dots,k$), and suppose, without loss of generality, that $\pi$ begins and ends in $\what{A}$: then for any $0\leqslant j\leqslant (k-1)/2$, the edge $\epsilon_{2j+1}$ is an edge in $\what{A}$ of weight 
\begin{equation}
\omega(\epsilon_{2j+1})=\sum_{e\in E_{A}(s_{2j+1},t_{2j+1}} \omega(e)\nonumber
\end{equation}
Similarly, for any $1\leqslant j\leqslant (k-1)/2$, the edge $\epsilon_{2j}$ is in $\what{B}$ and of weight
\begin{equation}
\omega(\epsilon_{2j})=\sum_{e\in E_{B}(s_{2j},t_{2j})} \omega(e)\nonumber
\end{equation}
Then, the weight of $\pi$, which is the weight of $f_{0}$, is given by:
\begin{equation}
\omega(\pi)=\prod_{1\leqslant i\leqslant k} \omega(\epsilon_{i})\nonumber
\end{equation}
The distribution of the product on the sum gives us that this is equal to:
\begin{eqnarray}
\omega(\pi)&=&\sum_{e_{1}\in E_{A}(s_{1},t_{1})}\sum_{e_{2}\in E_{B}(s_{2},t_{2})}\dots\sum_{e_{k}\in E_{A}(s_{k},t_{k})} \omega(e_{1})\omega(e_{2})\dots\omega(e_{k})\nonumber\\
&=&\sum_{\rho\in\path{s_{1},t_{k}}{A,B}} \omega(\rho)\nonumber\\
&=&\sum_{e\in E_{A\plug B}(s_{1},t_{k})} \omega(e)\nonumber
\end{eqnarray}
The last right-hand term is by definition the weight of the edge $f_{1}$ of $\what{A\plug B}$ which begins at $s_{1}$ and ends at $t_{k}$. Now, it is obvious that there is an edge between two vertices in $\what{\what{A}\plug\what{B}}$ if and only if there is an edge between these same vertices in $\what{A\plug B}$. Hence, since the weights of their corresponding edges are equal, the graphs are equal.
\end{proof}

\begin{theorem}[Compositionality]\label{compositiontruth}~\\
If $\de{f}$ and $\de{a}$ are successful projects in $\cond{A\multimap B}$ and $\cond{A}$ respectively, then the application $\de{f}\plug\de{a}$ is itself a successful project.
\end{theorem}

\begin{proof}
Let $\de{f}=(0,F)$ and $\de{a}=(0,A)$ be the two successful projects, and define $\what{f}=(0,\what{F})$ and $\what{a}=(0,\what{A})$. We show that $\de{b}=(b,B)$, the result of the reduction $\what{\de{f}}\plug\what{\de{a}}$, is indeed successful. As the reduction of two symmetric graphs, $B$ is symmetric.
The fact that $\what{A}$ is symmetric and satisfies $\what{A}^3=\what{A}$ implies that a given vertex cannot be the target or the source of more than one edge (see the proof of Proposition \ref{characterization} for details). Since all this is also true for $\what{F}$, it is clear that each vertex in $B$ is the source (resp. the target) of at most one edge, and this implies, combined with the fact that $B$ is symmetric, that $\what{B}^3=B=\what{B}$. The fact that $Tr(B)=0$ is an easy consequence of the fact that $Tr(A)=0=Tr(F)$.

The only remaining point is the question of the wager. Since all weights appearing in $F$ and $A$ are equal to $1$, we have that $b=\sca{a}{f}=0$ or $b=\sca{a}{f}=\infty$. But since the wager of a project cannot be equal to $\infty$, we have that $b=0$, hence $\de{b}$ is successful. This implies, using Lemma \ref{appli}, that $\de{f}\plug\de{a}$ is successful.
\end{proof}

Now that a notion of truth has been defined, it is quite natural to wonder wether a soundness and completeness theorem holds. While a soundness theorem will be given in the next section, we will not answer the question of the existence of a completeness theorem in this paper. However, the following result of (partial) internal completeness for the tensor product can be seen as a first step towards a positive answer.

\begin{proposition}\label{internalcomp}
If $\de{f}\in\cond{A}\otimes\cond{B}$ is a successful project, then there exists successful projects $\de{a}\in\cond{A}$ and $\de{b}\in\cond{B}$ such that $\de{f}=\de{a}\otimes\de{b}$.
\end{proposition}

\begin{proof}
Write $\de{f}=(0,F)$. We only need to show that $\what{F}$ can be written as the tensor product of two projects in $\cond{A}$ and $\cond{B}$. If it is not the case, there exists in $\what{F}$ an edge from a vertex $v$ in $V_{A}$ to a vertex $w$ in $V_{B}$. Now, consider the project $\de{c}=(c,C)$ where $c\not= 0$ and $C=(V_{A}\cup V_{B},\{(w,v)\},\omega((w,v))=1)$. This project is orthogonal to any element of $\cond{A\odot B}$, hence it is orthogonal to any element of $\cond{A\otimes B}$. However, we have that $\sca{f}{c}=\infty$ (since $\de{f}$ is successful, every edge in $\what{F}$ has weight $1$ by Proposition \ref{characterization}, and $\scalar{F,C}=\scalar{\what{F},C}$), so it is not orthogonal to $\de{f}$, which is contradictory.
\end{proof}

\section{Interpretation of Proofs}

In this section, we give the explicit interpretation of proofs of multiplicative linear logic with mix and units. Even though the results of this section were already obtained (for the most part) by defining the categorical model in section \ref{denot}, we believe this construction can help gaining a better understanding of our model, and acquire better insights on locativity.

Let us denote by $\delta$ the bijection $\mathbb{N}\times\mathbb{N}\rightarrow \mathbb{N}$ defined as $(n,m)\mapsto 2^{n}(2m+1)-1$. In this section, we will work up to the bijection $\delta$ and consider graphs whose set of vertices is a finite subset of $\mathbb{N}\times\mathbb{N}$.

\begin{definition}
We fix $\mathcal{V}ar=\{X_{i}(j)\}_{i,j\in\mathbb{N}}$ a set of \emph{localized variables\footnotemark\addtocounter{footnote}{-1}}. For $i\in\mathbb{N}$, the set $X_{i}=\{X_{i}(j)\}_{j\in\mathbb{N}}$ is said to be the \emph{variable name $X_{i}$}, and we call an element of $X_{i}$ a \emph{localized variable of name\footnote{The variable names are the variables in the usual sense (styled spiritual by Girard in ludics \cite{locussolum}), whereas the notion of localized variable is close to that of occurences.} $X_{i}$}. We suppose moreover that each name of variable $X_{i}$ comes with a \emph{size}\footnote{The size allows us to interpret atoms as complex conducts and not only by conduct of carrier of cardinality $1$.}, represented by an integer $n_{i}$. 
\end{definition}

For $i,j\in\mathbb{N}$ we define the \emph{location} $\sharp X_{i}(j)$ of the variable $X_{i}(j)$ as the set $\{(i,m)~|~jn_{i}\leqslant m\leqslant (j+1)n_{i}-1\}$.

\begin{definition}[Formulas of locMLL]
We inductively define the formulas of \emph{localized multiplicative linear logic} locMLL and their \emph{location} as follows:
\begin{itemize}
\item A localized variable $X_{i}(j)$ of name $X_{i}$ is a formula whose location is defined as $\sharp X_{i}(j)$;
\item If $X_{i}(j)$ is a localized variable of name $X_{i}$, then $(X_{i}(j))^{\pol}$ is a formula of location $\sharp X_{i}(j)$.
\item If $A,B$ are formulas of locations $X,Y$ such that $X\cap Y=\emptyset$, then $A\otimes B$ (resp. $A\parr B$) is a formula of location $X\cup Y$;
\item The constants $\cond{1}$ and $\bot$ are formulas of location $\emptyset$.
\end{itemize}
If $A$ is a formula, we will denote by $\sharp A$ the location of $A$. We also define sequents $\vdash \Gamma$ of locMLL when the formulas of $\Gamma$ have pairwise disjoint locations\footnote{This is a natural condition, since the comma in the sequent is interpreted as a $\parr$.}.
\end{definition}

\begin{definition}[Formulas of MLL+MIX]
We define the formulas of MLL by the following grammar:
\begin{equation}
F:=X_{i}~|~X_{i}^{\pol}~|~F\otimes F~|~F\parr F~|~\bot~|~\cond{1}\nonumber
\end{equation}
where $X_{i}$ is a variable name.
\end{definition}

\begin{remark}
In both locMLL and MLL+MIX, the negation of a composed formula is defined by using De Morgan's identities to push down the negation on atoms.
\end{remark}

\begin{remark}
Notice that to any locMLL formula there corresponds a formula of MLL obtained by simply replacing the variables by their names, i.e. applying $X_{i}(j)\mapsto X_{i}$. Conversely, we can \emph{localize} any MLL formula: if $e$ is an enumeration of the occurences of variable names in $\vdash \Gamma$, we can define a locMLL sequent $\vdash \Gamma^{e}$. For instance, the MLL formula $X_{1}\otimes (X_{1}\otimes X_{2})\multimap (X_{1}\otimes X_{2})\otimes X_{1}$ can be localized as $X_{1}(1)\otimes (X_{1}(2)\otimes X_{2}(1))\multimap (X_{1}(3)\otimes X_{2}(2))\otimes X_{1}(4)$, or as $X_{1}(42)\otimes (X_{1}(78)\otimes X_{2}(7))\multimap (X_{1}(99)\otimes X_{2}(88))\otimes X_{1}(1324)$,etc.
\end{remark}

\begin{definition}[Proofs of locMLL]
A proof of locMLL is a proof obtained using the sequent calculus rules of Figure \ref{locMLL}, and such that every variable $X_{i}(j)$ and every negation of variable $(X_{i}(j))^{\pol}$ appear in at most one axiom rule.
\end{definition}

\begin{definition}[Proofs of MLL+MIX]
A proof of MLL+MIX is a proof obtained using the sequent calculus rules of Figure \ref{MLL}.
\end{definition}

\begin{remark}
To each locMLL proof corresponds a MLL proof by replacing each localized variable in the proof by its name. Conversely, being given an enumeration $e$ of the occurences of variable names in the axiom rules of a MLL+MIX proof $\pi$, we can spread the enumeration to the whole derivation tree to obtain a locMLL proof $\pi^{e}$. For instance, the proof of Figure \ref{mllproof} can be localized as one of the proofs of Figure \ref{mlllocal}.
\end{remark}

\begin{figure}
\begin{center}
$\begin{array}{cc}
\begin{minipage}{5cm}
\begin{prooftree}
\AxiomC{}
\RightLabel{\scriptsize{Ax $(j\neq j')$}}
\UnaryInfC{$\vdash X_{i}(j)^{\pol},X_{i}(j')$}
\end{prooftree}
\end{minipage}
&
\begin{minipage}{5cm}
\begin{prooftree}
\AxiomC{$\vdash A,\Delta$}
\AxiomC{$\vdash A^{\pol},\Gamma$}
\RightLabel{\scriptsize{Cut\footnotemark\addtocounter{footnote}{-1}}}
\BinaryInfC{$\vdash \Delta,\Gamma$}
\end{prooftree}
\end{minipage}
\\
\begin{minipage}{5,5cm}
\begin{prooftree}
\AxiomC{$\vdash A,\Delta$}
\AxiomC{$\vdash B,\Gamma$}
\RightLabel{\scriptsize{$\otimes\footnotemark\addtocounter{footnote}{-1}$}}
\BinaryInfC{$\vdash A\otimes B,\Delta,\Gamma$}
\end{prooftree}
\end{minipage}
&
\begin{minipage}{5cm}
\begin{prooftree}
\AxiomC{}
\RightLabel{\scriptsize{$\cond{1}$}}
\UnaryInfC{$\vdash \cond{1}$}
\end{prooftree}
\end{minipage}
\\
\begin{minipage}{5cm}
\begin{prooftree}
\AxiomC{$\vdash A,B,\Gamma$}
\RightLabel{\scriptsize{$\parr$}}
\UnaryInfC{$\vdash A\parr B,\Gamma$}
\end{prooftree}
\end{minipage}
&
\begin{minipage}{5cm}
\begin{prooftree}
\AxiomC{$\vdash \Gamma$}
\RightLabel{\scriptsize{$\bot$}}
\UnaryInfC{$\vdash \bot,\Gamma$}
\end{prooftree}
\end{minipage}
\\
\multicolumn{2}{c}{
\begin{minipage}{5cm}
\begin{prooftree}
\AxiomC{$\vdash \Gamma$}
\AxiomC{$\vdash \Delta$}
\RightLabel{\scriptsize{mix\footnotemark}}
\BinaryInfC{$\vdash \Gamma,\Delta$}
\end{prooftree}
\end{minipage}
}
\end{array}
$\end{center}
\caption{Localized sequent calculus locMLL}\label{locMLL}
\end{figure}
\footnotetext{We need $(\sharp A\cup\sharp\Delta)\cap (\sharp B\cup\sharp\Gamma)=\emptyset$ to apply the $\otimes$ rule and $\sharp\Delta\cap\sharp\Gamma=\emptyset$ to apply the cut and mix rules.}
\begin{figure}
\begin{center}
$\begin{array}{cc}
\begin{minipage}{5cm}
\begin{prooftree}
\AxiomC{}
\RightLabel{\scriptsize{Ax}}
\UnaryInfC{$\vdash X_{i}^{\pol},X_{i}$}
\end{prooftree}
\end{minipage}
&
\begin{minipage}{5cm}
\begin{prooftree}
\AxiomC{$\vdash A,\Delta$}
\AxiomC{$\vdash A^{\pol},\Gamma$}
\RightLabel{\scriptsize{Cut}}
\BinaryInfC{$\vdash \Delta,\Gamma$}
\end{prooftree}
\end{minipage}
\\
\begin{minipage}{5,5cm}
\begin{prooftree}
\AxiomC{$\vdash A,\Delta$}
\AxiomC{$\vdash B,\Gamma$}
\RightLabel{\scriptsize{$\otimes$}}
\BinaryInfC{$\vdash A\otimes B,\Delta,\Gamma$}
\end{prooftree}
\end{minipage}
&
\begin{minipage}{5cm}
\begin{prooftree}
\AxiomC{}
\RightLabel{\scriptsize{$\cond{1}$}}
\UnaryInfC{$\vdash \cond{1}$}
\end{prooftree}
\end{minipage}
\\
\begin{minipage}{5cm}
\begin{prooftree}
\AxiomC{$\vdash A,B,\Gamma$}
\RightLabel{\scriptsize{$\parr$}}
\UnaryInfC{$\vdash A\parr B,\Gamma$}
\end{prooftree}
\end{minipage}
&
\begin{minipage}{5cm}
\begin{prooftree}
\AxiomC{$\vdash \Gamma$}
\RightLabel{\scriptsize{$\bot$}}
\UnaryInfC{$\vdash \bot,\Gamma$}
\end{prooftree}
\end{minipage}
\\
\multicolumn{2}{c}{
\begin{minipage}{5cm}
\begin{prooftree}
\AxiomC{$\vdash \Gamma$}
\AxiomC{$\vdash \Delta$}
\RightLabel{\scriptsize{mix}}
\BinaryInfC{$\vdash \Gamma,\Delta$}
\end{prooftree}
\end{minipage}
}
\end{array}
$\end{center}
\caption{Sequent calculus of MLL+MIX}\label{MLL}
\end{figure}

\begin{figure}
\begin{prooftree}
\AxiomC{}
\RightLabel{\scriptsize{Ax}}
\UnaryInfC{$\vdash X_{1},X_{1}^{\pol}$}
\AxiomC{}
\RightLabel{\scriptsize{Ax}}
\UnaryInfC{$\vdash X_{1},X_{1}^{\pol}$}
\RightLabel{\scriptsize{Cut}}
\BinaryInfC{$\vdash X_{1},X_{1}^{\pol}$}
\AxiomC{}
\RightLabel{\scriptsize{Ax}}
\UnaryInfC{$\vdash X_{2},X_{2}^{\pol}$}
\RightLabel{\scriptsize{$\otimes$}}
\BinaryInfC{$\vdash X_{1}^{\pol},X_{2}^{\pol},X_{1}\otimes X_{2}$}
\RightLabel{\scriptsize{$\parr$}}
\UnaryInfC{$\vdash X_{1}^{\pol}\parr X_{2}^{\pol},X_{1}\otimes X_{2}$}
\end{prooftree}
\caption{A proof of MLL+MIX}\label{mllproof}
\end{figure}

\begin{figure}
\begin{prooftree}
\AxiomC{}
\RightLabel{\scriptsize{Ax}}
\UnaryInfC{$\vdash X_{1}(3),X_{1}(97)^{\pol}$}
\AxiomC{}
\RightLabel{\scriptsize{Ax}}
\UnaryInfC{$\vdash X_{1}(97),X_{1}(23)^{\pol}$}
\RightLabel{\scriptsize{Cut}}
\BinaryInfC{$\vdash X_{1}(3),X_{1}(23)^{\pol}$}
\AxiomC{}
\RightLabel{\scriptsize{Ax}}
\UnaryInfC{$\vdash X_{2}(7),X_{2}(14)^{\pol}$}
\RightLabel{\scriptsize{$\otimes$}}
\BinaryInfC{$\vdash X_{1}(23)^{\pol},X_{2}(12)^{\pol},X_{1}(3)\otimes X_{2}(7)$}
\RightLabel{\scriptsize{$\parr$}}
\UnaryInfC{$\vdash X_{1}(23)^{\pol}\parr X_{2}(12)^{\pol},X_{1}(3)\otimes X_{2}(7)$}
\end{prooftree}
\caption{A proof of locMLL+MIX}\label{mlllocal}
\end{figure}

\begin{definition}[Interpretations]
We define a \emph{basis of interpretation} as a function $\Phi$ which associates to each variable name $X_{i}$ a conduct of carrier $\{0,\dots,n_{i}-1\}$.
\end{definition}

\begin{definition}[Interpretation of formulas of locMLL]
Let $\Phi$ be a basis of interpretation. We define the interpretation $I_{\Phi}(F)$ along $\Phi$ of a formula $F$ inductively:
\begin{itemize}
\item If $F=X_{i}(j)$, then $I_{\Phi}(F)$ is the delocation (i.e. a conduct) of $\Phi(X_{i})$ along the bijection $x\mapsto (i,jn_{i}+x)$;
\item If $F=(X_{i}(j))^{\pol}$, we define the conduct $I_{\Phi}(F)=(I_{\Phi}(X_{i}(j)))^{\pol}$;
\item If $F=\cond{1}$ (resp. $F=\bot$), we define $I_{\Phi}(F)$ as the conduct $\cond{1}$ (resp. $\bot$);
\item If $F=A\otimes B$, we define the conduct $I_{\Phi}(F)=I_{\Phi}(A)\otimes I_{\Phi}(B)$;
\item If $F=A\parr B$, we define the conduct $I_{\Phi}(F)=I_{\Phi}(A)\parr I_{\Phi}(B)$.
\end{itemize}
Moreover, a sequent $\vdash \Gamma$ will be interpreted as the $\parr$ of the formulas of $\Gamma$, which we will denote by $\bigparr \Gamma$.
\end{definition}

\begin{definition}[Interpretation of proofs of locMLL]\label{interpretproofs}
Let $\Phi$ be a basis of interpretation. We define the interpretation of a proof (a project) $I_{\Phi}(\pi)$ inductively as follows:
\begin{itemize}
\item if $\pi$ consists solely of an axiom rule introducing $\vdash (X_{i}(j))^{\pol},X_{i}(j')$, we define $I_{\Phi}(\pi)$ as the $\de{Fax}$ obtained from the bijection $(i,jn_{i}+x)\mapsto (i,j'n_{i}+x)$;
\item if $\pi$ consists solely of a $\cond{1}$ rule, we define $I_{\Phi}(\cond{1})=(0,0)$, where $0$ denotes the empty graph on an empty set of vertices;
\item if $\pi$ is obtained from $\pi'$ by a $\parr$ rule, then $I_{\Phi}(\pi)=I_{\Phi}(\pi')$;
\item if $\pi$ is obtained from $\pi_{1}$ and $\pi_{2}$ by a $\otimes$ rule, we define $I_{\Phi}(\pi)=I_{\Phi}(\pi_{1})\otimes I_{\Phi}(\pi')$;
\item if $\pi$ is obtained from $\pi_{1}$ and $\pi_{2}$ by a cut rule, we define $I_{\Phi}(\pi)=I_{\Phi}(\pi_{1})\plug I_{\Phi}(\pi_{2})$;
\item if $\pi$ is obtained from $\pi_{1}$ and $\pi_{2}$ by a mix rule, we define $I_{\Phi}(\pi)=I_{\Phi}(\pi_{1})\otimes I_{\Phi}(\pi_{2})$;
\item if $\pi$ is obtained from $\pi'$ by a $\bot$ rule, we define $I_{\Phi}(\pi)=I_{\Phi}(\pi')$.
\end{itemize}
\end{definition}

\begin{proposition}[Full localized soundness]\label{locsound}
Let $\Phi$ be a basis of interpretation. If $\pi$ is a proof of conclusion $\vdash \Delta$, then $I_{\Phi}(\pi)$ is a successful project in the conduct $I_{\Phi}(\vdash\Delta)$.
\end{proposition}

\begin{proof}
We prove it by induction on the last rule of $\pi$. By definition, the interpretation of the axiom rule introducing $\vdash (X_{i}(j))^{\pol},X_{i}(j')$ gives a successful project in $I_{\Phi}(X_{i}(j))\multimap I_{\Phi}(X_{i}(j'))$ which is equal to $I_{\Phi}((X_{i}(j))^{\pol}\parr X_{i}(j'))$.
Then:
\begin{itemize}
\item if $\pi$ is the $\cond{1}$ rule, then $\pi=(0,0)$ is successful and in $\cond{1}$;
\item if the last rule of $\pi$ is a $\otimes$ rule between $\pi_{1}$ and $\pi_{2}$ with $\pi_{i}$ of conclusion $\vdash A_{i},\Gamma_{i}$, then $\pi=\pi_{1}\otimes\pi_{2}$, which is a successful project in $(A_{1}\parr (\bigparr\Gamma_{1}))\otimes(A_{2}\parr (\bigparr \Gamma_{2}))\subseteq(A_{1}\otimes A_{2})\parr (\bigparr \Gamma)$;
\item if the last rule of $\pi$ is a $\parr$ rule, then $I_{\Phi}(\pi)\in I_{\Phi}(A_{1}\parr A_{2}\parr(\bigparr \Gamma))$ by definition;
\item if $\pi$ ends with a $\bot$ rule on $\pi'$, the interpretation of $\pi$ is the same as the interpretation of $\pi'$, and the interpretation of the formula $\bigparr \Gamma$ is equal to the interpretation of the formula $\bot\parr(\bigparr \Gamma)$ since $\bot$ is the unit of $\parr$;
\item if $\pi$ is obtained through a cut rule between $\pi_{1}$ and $\pi_{2}$, of respective conclusions $\vdash A,\Gamma_{1}$ and $\vdash A^{\pol},\Gamma_{2}$, then Theorem \ref{compositiontruth} tells us that $I_{\Phi}(\pi_{1})\plug I_{\Phi}(\pi_{2})$ is a successful project in $\bigparr \Gamma$;
\item if $\pi$ is obtained from $\pi_{1}$ and $\pi_{2}$ of respective conclusions $\vdash \Gamma_{1}$ and $\vdash \Gamma_{2}$ by a mix rule, we know that $I_{\Phi}(\pi)$ is a successful project in $(\bigparr \Gamma_{1})\otimes(\bigparr \Gamma_{2})$, which is included in the conduct $\bigparr \Gamma$ from Proposition \ref{mixrule}.
\end{itemize}
\end{proof}

\begin{theorem}[Full Soundness of MLL+MIX]
Let $\Phi$ be a basis of interpretation, $\pi$ a proof of MLL+MIX of conclusion $\vdash \Gamma$, and $e$ an enumeration of occurences of variables in the axioms of $\pi$. Then $I_{\Phi}(\pi^{e})$ is a successful project in $I_{\Phi}(\vdash \Gamma^{e})$.
\end{theorem}

\begin{proof} It is an immediate corollary of Proposition \ref{locsound}.
\end{proof}

\begin{lemma}\label{lemcutelim}
If $\de{a}_{i}$ ($i=1,2,3$) are projects, then:
\begin{equation}
(\de{a}_{1}\otimes \de{a_{2}})\plug \de{a}_{3}=(\de{a}_{1}\plug\de{a}_{3})\plug\de{a}_{2}\nonumber
\end{equation}
\end{lemma}

\begin{proof}
Let $\de{a}_{i}=(a_{i},A_{i})$ be projects. First, we notice that $(A_{1}\cup A_{2})\plug A_{3}=(A_{1}\plug A_{3})\plug A_{2}$. Indeed, both graphs are defined on the same set of vertices, and moreover there is a one-to-one function (preserving weights) between their sets of edges: an edge $\{e_{i}\}_{0\leqslant i\leqslant n}$ in the graph $(A_{1}\plug A_{3})\plug A_{2}$ is an alternation of edges in $A_{2}$ and paths alternating between $A_{1}$ and $A_{3}$, and therefore corresponds to one (and exactly one) path alternating between $A_{3}$ and $A_{1}\cup A_{2}$. Then, using the adjunction (we write $\bar{a}=a_{1}+a_{2}+a_{3}$):
\begin{eqnarray}
(\de{a}_{1}\otimes\de{a}_{2})\plug\de{a}_{3}&=&(\bar{a}+\scalar{A_{1}\cup A_{2},A_{3}}, (A_{1}\cup A_{2})\plug A_{3})\nonumber\\
&=&(\bar{a}+\scalar{A_{1},A_{3}}+\scalar{A_{1}\plug A_{3},A_{2}}, (A_{1}\plug A_{3})\plug A_{2})\nonumber\\
&=&(a_{1}+a_{3}+\scalar{A_{1},A_{3}},A_{1}\plug A_{3})\plug (a_{2},A_{2})\nonumber\\
&=&((a_{1},A_{1})\plug (a_{2},A_{2}))\plug (a_{3},A_{3})\nonumber
\end{eqnarray}
\end{proof}

\begin{corollary}\label{commut}
If $\de{a}_{i}$ ($i=1,2,3$) are projects, where $\de{a}_{1}$ and $\de{a}_{2}$ are of disjoint carriers and $\de{a}_{1}$ and $\de{a}_{3}$ are also of disjoint carriers, then:
\begin{equation}
(\de{a}_{1}\otimes \de{a_{2}})\plug \de{a}_{3}=\de{a}_{1}\otimes(\de{a}_{2}\plug\de{a}_{3})\nonumber
\end{equation}
\end{corollary}

\begin{proof}
From the preceding lemma, we have:
\begin{equation}
(\de{a}_{1}\otimes \de{a_{2}})\plug \de{a}_{3}=\de{a}_{1}\plug(\de{a}_{2}\plug\de{a}_{3})\nonumber
\end{equation}
Since $\de{a}_{1}$ and $\de{a}_{3}$ are of disjoint carriers, the carriers of $\de{a}_{1}$ and $\de{a}_{2}\plug\de{a}_{3}$ are disjoint ($\de{a}_{1}$ and $\de{a}_{2}$ are of disjoint carrier since their tensor product is defined). Therefore, we have $\de{a_{1}}\plug(\de{a}_{2}\plug\de{a}_{3})=\de{a}_{1}\otimes(\de{a}_{2}\plug\de{a}_{3})$.
\end{proof}

\begin{theorem}[Invariance by cut-elimination]
Let $\Phi$ be a basis of interpretation. If $\pi$ is proof of locMLL and $\pi'$ is the cut-free proof obtained by eliminating the cuts\footnote{We do not define the cut elimination procedure since it is the same as usual, and the fact that we are localized changes nothing.} in $\pi$, then $I_{\Phi}(\pi)=I_{\Phi}(\pi')$.
\end{theorem}

\begin{proof}
We show that interpretation is preserved through every steps of the cut-elimination procedure:
\begin{itemize}
\item if $\pi$ is a cut between two axioms introducing $\vdash (X_{i}(j))^{\pol},X_{i}(j')$ and $\vdash (X_{i}(j'))^{\pol},X_{i}(j'')$ then $I_{\Phi}(\pi)=\de{Fax}_{1}\plug \de{Fax}_{2}$ where $\de{Fax}_{1}$ and $\de{Fax}_{2}$ are given by Definition \ref{interpretproofs}; we easily verify that the reduction of two faxes is a $\de{Fax}$: here it is the one we obtain from the bijection $(i,jn_{i}+x)\mapsto (i,j''n_{i}+x)$ for $0\leqslant x\leqslant n_{i}-1$, i.e. the interpretation of the axiom rule introducing $\vdash (X_{i}(j))^{\pol},X_{i}(j'')$, result of the cut elimination applied to $\pi$:
\begin{prooftree}
\AxiomC{}
\RightLabel{\scriptsize{Ax ($j\neq j''$)}}
\UnaryInfC{$\vdash (X_{i}(j))^{\pol},X_{i}(j'')$}
\end{prooftree}
\item if $\pi$ is a cut between two proofs, one obtained from a tensor rule between $\pi_{1}$ and $\pi_{2}$, and the other obtained from a $\parr$ rule on $\pi_{3}$:
\begin{prooftree}
\AxiomC{$\vdots^{\pi_{1}}$}
\noLine
\UnaryInfC{$\vdash \Delta_{1},A$}
\AxiomC{$\vdots^{\pi_{2}}$}
\noLine
\UnaryInfC{$\vdash \Delta_{2},B$}
\RightLabel{\scriptsize{$\otimes$}}
\BinaryInfC{$\vdash \Delta_{1},\Delta_{2},A\otimes B$}
\AxiomC{$\vdots^{\pi_{3}}$}
\noLine
\UnaryInfC{$\vdash \Delta_{3},A^{\pol},B^{\pol}$}
\RightLabel{\scriptsize{$\parr$}}
\UnaryInfC{$\vdash \Delta_{3},A^{\pol}\parr B^{\pol}$}
\RightLabel{\scriptsize{cut}}
\BinaryInfC{$\vdash \Delta_{1},\Delta_{2},\Delta_{3}$}
\end{prooftree}
We have, denoting by $\de{a}_{i}=I_{\Phi}(\pi_{i})$ ($i=1,2,3$), $I_{\Phi}(\pi)=(\de{a}_{1}\otimes\de{a}_{2})\plug\de{a}_{3}$, which is equal to $(\de{a}_{3}\plug\de{a}_{2})\plug\de{a}_{1}$ by Lemma \ref{lemcutelim}, which is the interpretation of the proof:
\begin{prooftree}
\AxiomC{$\vdots^{\pi_{1}}$}
\noLine
\UnaryInfC{$\vdash \Delta_{1},A$}
\AxiomC{$\vdots^{\pi_{3}}$}
\noLine
\UnaryInfC{$\vdash \Delta_{3},A^{\pol},B^{\pol}$}
\RightLabel{\scriptsize{cut}}
\BinaryInfC{$\vdash \Delta_{1},\Delta_{3},B^{\pol}$}
\AxiomC{$\vdots^{\pi_{2}}$}
\noLine
\UnaryInfC{$\vdash \Delta_{2},B$}
\RightLabel{\scriptsize{cut}}
\BinaryInfC{$\vdash \Delta_{1},\Delta_{3},\Delta_{2}$}
\end{prooftree}
\item the commutation rules are clear from Corollary \ref{commut}.
\end{itemize}
These are the only cases, since we considered a sequent calculus with only atomic axioms.
\end{proof}

The question of the completeness of our model is still open, and will be the object of a future work.

\section{Adjacency Matrices}\label{matrix}

In this section, we will show some results that explain some connections between the operations we defined on graphs and operator-theoretic notions. These results will allow us to show (in the next section) that our framework on graphs, when restricted to a certain class of graphs, is connected to Girard's geometry of interaction \cite{goi5}. In the remaining two sections (and particularly in the next), we will use some operator-theoretic notions that may not be familiar to the reader. While we are not able to write a complete introduction to such matters, we though useful to compile a list of important results and definitions in the appendix \ref{appendix}.

Since objects in \cite{goi5} are hermitian operators of norm $\leqslant 1$, we will restrict to a certain class of graphs that correspond to hermitian matrices of norm $\leqslant 1$. We then show that the different definitions we gave on graphs can be translated into linear algebraic definitions. In particular, we can prove that the adjunction is still valid, which implies that this restriction defines a geometry of interaction in the same way we defined our GoI in section \ref{goi}. Moreover, the linear algebraic definitions that correspond to our definitions on graphs are exactly the same as GoI5 definitions, as we will show in the next section.

Let $\hil{H}$ be an (countable) infinite-dimensional Hilbert space. We fix an orthonormal basis $(e_{i})_{i\in\mathbb{N}}$ of the Hilbert space $\hil{H}$. Given a finite subset $S\subset\mathbb{N}$, there is a projection on the subspace generated by $\{e_{i}~|~i\in S\}$ that we will denote by $p_{S}$. Then the restriction $p_{S}\mathcal{B}(\hil{H})p_{S}$ is isomorphic to the algebra of $n\times n$ matrices $\mathcal{M}_{n}(\mathbb{C})$ where $n$ is the cardinal of $S$. All graphs we consider in this section and the following are such that their set of vertices is a finite subset of $\mathbb{N}$.

\begin{definition}[Localized adjacency matrix]
If $G$ is a simple weighted graph, the adjacency matrix (the matrix of weights) $\mat{G}$ of $G$ defines an operator in $p_{V_{G}}\mathcal{B}(\hil{H})p_{V_{G}}$ (hence in $\mathcal{B}(\hil{H})$) whose matrix is $\mat{G}$ in the base $\{e_{i}\}_{i\in V_{G}}$. We will make an abuse of notations and denote this operator, the \emph{localized adjacency matrix of $G$}, by $\mat{G}$. 
\end{definition}

\begin{definition}[Operator Graph]
We will call \emph{operator graph} a simple symmetric weighted graph $G$ such that $\lVert\mat{G}\rVert \leqslant 1$.
\end{definition}

We recall that if $G,H$ are graphs on the same set of vertices, then the product $\mat{G}\mat{H}$ is the adjacency matrix of the graph of paths of length $2$ with the first edge in $G$ and the second in $H$. This is the key ingredient for the following propositions.

\begin{proposition}\label{measures}
Let $F,G$ be operator graphs, $\mat{F}$ and $\mat{G}$ their localized adjacency matrices. The product of $\mat{F}$ and $\mat{G}$ as elements of $\mathcal{B}(\hil{H})$ gives an operator in $(p_{V_{F}\cup V_{G}})\mathcal{B}(\hil{H})(p_{V_{F}\cup V_{G}})$ and:
\begin{equation}
\scalar{F,G}= \sum_{k=1}^{\infty} \frac{Tr((\mat{F}\mat{G})^{k})}{k}\nonumber
\end{equation}
\end{proposition}

\begin{proof}
Let $T_{n}=tr((\mat{F}\mat{G})^{n})/n$. We recall that the diagonal coefficient $\delta_{i}$ of $(\mat{F})^{n}$ is equal to the sum of the weights of the cycles of length $n$ in $F$ that begin and end at $i$, and each path is counted exactly once. This means that in $tr((\mat{F}\mat{G})^{n})$ each alternating circuit $\bar{\rho}$ is counted exactly $\sharp\bar{\rho}$ times, where $\sharp\bar{\rho}$ is the cardinality of the set $\bar{\rho}$ defined in Proposition \ref{countingeqclas}. Thus $T_{n}$ is equal to the sum, for all alternating circuits $\bar{\rho}$ of length $n$ in $F\bicol G$, of $\sharp \bar{\rho}.\omega_{F\bicol G}(\rho)/n$. We then have that $T_{n}$ is equal to the sum, for all $d$-circuits $\pi=\rho_{\pi}^{d}$ of length $n$, of the terms $\omega_{F\bicol G}(\pi)/d=\omega_{F\bicol G}(\rho_{\pi})^{d}/d$ (recall that $\bar{\rho_{\pi}}$ is of cardinality $n/d$).

Let us now choose a $1$-circuit $\bar{\pi}$ of length $k$. We have just seen that each term $\Omega^{\bar{\pi}}_{d}=\omega_{F\bicol H}(\bar{\pi})^{d}/d$ appears in $\sum_{n=1}^{\infty} T_{n}$ and it appears only once (in the term $T_{dk}$). Summing these terms, we obtain $-log(1-\omega(\bar{\pi}))$. Eventually, by taking the sum over all $1$-circuits $\bar{\pi}\in\cycl{F,G}$, we obtain $\scalar{F,G}$.
\end{proof}

In the following, when working with complex logarithms, we will always consider the principal branch of the logarithm.

\begin{lemma}
Let $a$ be a square matrix such that $\norm{a}\leqslant 1$. Then, with the convention that $-log(0)=\infty$, $$-log(det(1-a))=\sum_{k=1}^{\infty} Tr(a^{k})/k$$
\end{lemma}

\begin{proof}
First notice that we have $-log(det(1-a))=Tr(-log(1-a))$, since\footnote{This formula follows from the equality $det(exp(A))=exp(Tr(A))$ for any square matrix $A$, a formula easily shown by considering $A$ written as a triangular matrix.} $det(1-a)=exp(Tr(log(1-a)))$. 

We first suppose that $1$ is not an eigenvalue of $a$ and write $\lambda_{1},\dots,\lambda_{n}$ these eigenvalues. The (principal branch of the) cologarithm of $1-a$ is defined as the series $\sum_{k\geqslant 1} a^{k}/k$ which converges\footnote{This is a straigthforward application of Dedekind's test: $\sum a_{n}b_{n}$ is convergant if $\sum (b_{n}-b_{n+1})$ converges absolutely, $b_{n}\rightarrow 0$ and $\sum a_{n}$ has bounded partial sums (see for instance Knopp's "Theory and Application of Infinite Series" \cite{Knopp}).} for every complex number $a\neq 1$ such that $\abs{a}\leqslant 1$. The matrix $1-a$ being invertible, the logarithm $-log(1-a)$ exists and its eigenvalues are equal to $\sum_{k\geqslant 1} \lambda_{i}^{k}/k=-log(1-\lambda_{i})$. Then we have:
\begin{equation}
Tr(-log(1-a))=\sum_{i=1}^{n}\sum_{k\geqslant 1} \frac{\lambda_{i}^{k}}{k}=\sum_{k\geqslant 1}\frac{Tr(a^{k})}{k}\nonumber
\end{equation}

Let us now suppose that $\lambda_{1}=1$. We rewrite the sum $\sum_{k\geqslant 1}Tr(a^{k})/k$ as $\sum_{k=1}^{\infty} \sum_{i=1}^{n} \lambda_{i}^{k}/k=\sum_{k=1}^{\infty} \sum_{i=1}^{n} \lambda_{i}^{2k}(1/2k+\lambda_{i}/(2k+1))$. This is equal to $\sum_{k=1}^{\infty} [1/2k+1/(2k+1)+\sum_{i=2}^{n}  \lambda_{i}^{2k}(1/2k+\lambda_{i}/(2k+1))]$ which is greater than $\sum_{k=1}^{\infty} 1/2k+1/(2k+1)$. This last series being divergent, we conclude that $\sum_{k=1}^{\infty} Tr(a^{k})/k=\infty$. But, since $1$ is an eigenvalue of $a$, the kernel of $1-a$ is non trivial, hence $det(1-a)=0$, which means that $-log(det(1-a))=\infty=\sum_{k=1}^{\infty} Tr(a^{k})/k$.
\end{proof}

\begin{corollary}\label{ldet}
Let $F,G$ be operator graphs. Then the product of $\mat{F}$ and $\mat{G}$ in $\mathcal{B}(\hil{H})$ gives an operator in $p_{V_{F}\cup V_{G}}\mathcal{B}(\hil{H})p_{V_{F}\cup V_{G}}$ and\footnote{The determinant is defined: since $V_{F}$ and $V_{G}$ are finite, $\mat{F}\mat{G}$ can be written as a square matrix.}:
\begin{equation}
\scalar{F,G}=-log(det(1-\mat{F}\mat{G}))\nonumber
\end{equation}
\end{corollary}

\begin{proof}
This is a direct consequence of Proposition \ref{measures} and the preceding lemma.
\end{proof}

\begin{proposition}
Let $F$ and $G$ be operator graphs. If $\scalar{F,G}\neq \infty$, then $\what{F\plug G}$ is total.
\end{proposition}

\begin{proof}
By definition, $\what{F\plug G}$ is total if and only if for all couple $v,v'$ of vertices in the symmetric difference $S=V_{F}\Delta V_{G}$ the following sum converges:
\begin{equation}
\sum_{\pi\in\path{v,v'}{F,G}} \omega_{F\bicol G}(\pi)\nonumber
\end{equation}
Let us fix $v,v'$ two vertices and denote by $E$ the set of alternating paths from $v$ to $v'$ in $F\bicol G$ that do not contain a cycle. Then, since $S$ is finite we know that $E$ is finite, and there is a path $\gamma$ of maximal weight. Then, we can say that
\begin{equation}
\sum_{\pi=(v,v')\in\path{}{F,G}} \omega_{F\bicol G}(\pi)\leqslant \sharp(E)\omega_{F\bicol G}(\gamma)\left(\sum_{\pi\in\cycl{F,G}}  \omega_{F\bicol G}(\pi)\right)\nonumber
\end{equation}
The right-hand of the equation being equal to $\sharp(E)\omega_{F\bicol G}(\gamma)\scalar{F,G}$, it is finite.
\end{proof}

\begin{proposition}\label{execution}
Suppose $F$ and $G$ are operator graphs, $\scalar{F,G}\neq\infty$. Then $\mat{H}=\mat{\what{F\plug G}}$ is the solution to the feedback equation\footnote{The feedback equation is the operator-theoretic counterpart to the cut-elimination procedure introduced and solved by Girard \cite{feedback}; it is explained and discussed in Girard's \emph{Blind Spot} \cite{blindspot} with a new and more elegant proof.} between $\mat{F}$ and $\mat{G}$, and therefore an operator graph.
\end{proposition}

\begin{proof}
By a similar argument to that of the preceding proof, we can show that for any couple of vertices $v,v'\in V_{F}\cup V_{G}$, the sum
\begin{equation}
\sum_{\pi\in\path{v,v'}{F,G}} \omega_{F\bicol G}(\pi)\nonumber
\end{equation}
is convergent. Supposing that $v,v'\in V_{F}\cap V_{G}$, and since $\omega_{F\bicol G}$ is always positive, it follows that the sum of $\omega_{F\bicol G}(\pi)$ over all paths $\pi$ that begin with an edge in $G$ and ends with an edge in $F$ is convergent. Which means that $\<\sum_{k=0}^{\infty} (\mat{F}\mat{G})^{k} e_{v},e_{v'}\>$ is convergent for all couple $v,v'\in V_{F}\cap V_{G}$. Hence  $1-\mat{G}\mat{F}$ is invertible, and the solution of the feedback equation is the hermitian of norm at most $1$ defined by
\begin{equation}
S=(p_{V_{F}'}\mat{F}+p_{V_{G}'})(1-\mat{G}\mat{F})^{-1}(\mat{G}p_{V_{G}'}+p_{V_{F}'})\nonumber
\end{equation}
It is a straightforward computation to show that $S=\mat{H}$. Let us write $V_{F}'=V_{F}-V_{G}$ and $V_{G}'=V_{G}-V_{F}$. The value $\omega_{H}((v,v'))=\<\mat{H}e_{v},e_{v'}\>=H_{v}^{v'}$ is given by
\begin{equation}
H_{v}^{v'}=\left\{\begin{array}{ll}
\<\sum_{k=0}^{\infty}(\mat{F}\mat{G})^{k}\mat{F}e_{v},e_{v'}\>&\text{for $v,v'\in V_{F}'$}\\
\<\mat{G}\sum_{k=0}^{\infty}(\mat{F}\mat{G})^{k}\mat{F}e_{v},e_{v'}\>&\text{for $v\in V_{F}',v'\in V_{G}'$}\\
\<\mat{F}\sum_{k=0}^{\infty}(\mat{G}\mat{F})^{k}\mat{G}e_{v},e_{v'}\>&\text{for $v\in V_{G}',v'\in V_{F}'$}\\
\<\sum_{k=0}^{\infty}(\mat{G}\mat{F})^{k}\mat{G}e_{v},e_{v'}\>&\text{for $v,v'\in V_{G}'$}
\end{array}\right.
\end{equation}
Thus, $\mat{H}$ is equal to $S$.
\end{proof}

\begin{proposition}[Adjunction]\label{adjunctmatrix}
Let $F,G_{1},G_{2}$ be operator graphs with $V_{F}=V_{G_{1}}\cup V_{G_{2}}$ and $V_{G_{1}}\cap V_{G_{2}}=\emptyset$. Suppose $H=\what{F\plug G_{1}}$ is total. We have the following adjunction.
\begin{equation}
\scalar{F,G_{1}\cup G_{2}}=\scalar{F,G_{1}}+\scalar{H,G_{2}}\nonumber
\end{equation}
\end{proposition}

\begin{proof}
This is a straightforward corollary of Proposition \ref{invariant} and the adjunction on graphs.
\end{proof}

\begin{remark}
We can define a restriction of our framework to operator graphs by replacing the composition of two graphs $F$ and $G$ by $\what{F\plug G}$. Then, all results of the previous sections hold, since the adjunction holds. This restricted version can moreover be rephrased by replacing graphs by matrices, since all our construction can be translated as operator-theoretical constructions. The model thus obtained, an intermediate framework between Girard's geometry of interaction in the hyperfinite factor and our own framework, is a finite-dimensional version of Girard's approach, as the following section will show.
\end{remark}

All the results of this section can be used in two different ways. They are of some importance in themselves, since — as we explained in the last remark — the restriction to operator projects gives rise to a “type $\text{I}$ geometry of interaction", i.e. a geometry of interaction whose principal objects are matrices. This geometry of interaction can be shown to have all properties of the graph framework defined in section \ref{goi}, and gives rise to a $\ast$-autonomous category  (see section \ref{denot}) and a notion of truth (section \ref{truth}) in the same way graphs do. We won't go any further in that direction since it is a simple adaptation of what we have done precedently. 

The second direction is given by the fact that from these results, one can define an embedding of operator projects into “hyperfinite projects", i.e. projects of Girard last geometry of interaction, and show that, through this embedding, the measurement between projects (i.e. the interaction), and the basic constructions on graphs corresponds to Girard's measurement and constructions.

\section{The Hyperfinite Factor}\label{hyp}

In this section, we will use notations and definitions of the geometry of interaction in the hyperfinite factor, that can be found in the original article \cite{goi5}. The reader will find an overview of the main used notions in the appendix. Our aim is to show how we can map operator projects (projects whose graph is an operator graph) to projects of GoI5 that preserves the measurement between projects $\sca{a}{b}$, and all the basic operations (tensor product, execution). We will first recall some definitions of Girard's GoI5, and the proceed to define the embedding and state the correspondence.

\subsection{Girard's GoI5}

We begin with the definition of the Fuglede-Kadison determinant, and a technical result concerning this determinant (Proposition \ref{tracepreserving}) that will be used in the next subsection. Then, we will make a quick overview of Girard's definitions.

In  any C$^{\ast}$-algebra, elements of the form $A^{\ast}A$ are called positive. Every positive element has a unique square root, and by analogy with the complex numbers, this square root is denoted by $\abs{A}$.

\begin{definition}[Fuglede-Kadison determinant]\label{FKdet}
Let $\mathcal{A}$ be a finite factor, and $T$ its normalized trace. Define, on the group of invertible operators, the Fuglede-Kadison determinant
\begin{equation}
\Delta(A)=exp(T(log(\abs{A})))\nonumber
\end{equation}
Then $\Delta$ can be extended\footnote{In the original article \cite{FKdet}, two extensions are considered: the "algebraic extension" and the "analytic extension". It is shown, however, that none of these extensions is continuous, and the term "extension" we use should not be confused with "extension by continuity".} to $\mathcal{A}$.
\end{definition}

\begin{remark} The Fuglede-Kadison determinant takes only positive values.
\end{remark}

We will also use the following lemma, which is proved in the original paper by Fuglede and Kadison \cite{FKdet}.

\begin{lemma}\label{FKnullspace}
Let $det_{FK}$ be any extension of the Fuglede-Kadison determinant to $\mathcal{A}$. If $u$ is an arbitrary operator with a non-trivial nullspace, $det_{FK}(u)=0$.
\end{lemma}

For the remaining definitions of this subsection, we consider we have chosen a trace $tr$ on the hyperfinite factor $\infhyp $ of type $\text{II}_{\infty}$ once and for all. Moreover, if $\alpha$ is a normal hermitian tracial form on a finite von Neumann algebra $\mathcal{A}$ and $p\in\infhyp $ is a finite projection, then $(p\infhyp  p)\otimes\mathcal{A}$ is a finite von Neumann algebra, and one can define on it the (Fuglede-Kadison) determinant as an extension of the following expression (defined on the group of invertible operators):
\begin{equation}
det^{p}_{tr\otimes\alpha}(A)=exp(tr_{\upharpoonright_{p\infhyp  p}}\otimes\alpha(log(\abs{A})))\nonumber
\end{equation}

\begin{definition}[Girard's project]
A \emph{Girard's project} will be a tuple $\de{a}=(p,a,\mathcal{A},\alpha,A)$ consisting of:
\begin{itemize}
\item a finite projection $p^{\ast}=p^{2}=p\in\infhyp $, the \emph{carrier} of the project $\de{a}$;
\item a real number (eventually infinite) $a\in\mathbf{R}\cup\{\infty\}$, the \emph{wager} of $\de{a}$;
\item a finite and hyperfinite von Neumann algebra $\mathcal{A}$, the \emph{dialect} of $\de{a}$;
\item a normal hermitian tracial form $\alpha$ on $\mathcal{A}$, the \emph{diatrace} of $\de{a}$;
\item a self-adjoint operator $A\in (p\infhyp  p)\otimes\mathcal{A}$ such that $\norm{A}\leqslant 1$.
\end{itemize}
As in Girard's paper, we will denote such an object by $\de{a}=a\cdot +\cdot \alpha+A$.
\end{definition}

For the following definitions, one needs to define two maps. Let $\mathcal{A,B}$ be finite von Neumann algebras, and $A,B$ be operators in respectively $\infhyp \otimes\mathcal{A}$ and $\infhyp \otimes\mathcal{B}$. We define $A^{\dagger_{\mathcal{B}}}$ and $B^{\ddagger_{\mathcal{A}}}$ through the following maps (defining $\tau:\mathcal{B}\otimes\mathcal{A}\rightarrow \mathcal{A}\otimes\mathcal{B}$ in the obvious way):
\begin{equation}
\begin{array}{lrclrcl}
(\cdot)^{\dagger_{\mathcal{B}}}: &\infhyp \otimes\mathcal{A}&\rightarrow&\infhyp \otimes\mathcal{A}\otimes\mathcal{B}, &A&\mapsto& A\otimes 1_{\mathcal{B}}\\
(\cdot)^{\ddagger_{\mathcal{A}}}: &\infhyp \otimes\mathcal{B}&\rightarrow&\infhyp \otimes\mathcal{A}\otimes\mathcal{B}, &B&\mapsto& (Id\otimes\tau)(B\otimes 1_{\mathcal{A}})
\end{array}\nonumber
\end{equation}

\begin{definition}[Orthogonality]
Let $\de{a}=a\cdot+\cdot\alpha+A$ and $\de{b}=b\cdot+\cdot\beta+B$ be two projects with the same carrier $p$. Define the measurement:
\begin{equation}
\sca{a}{b}=a\beta(1_{\mathcal{B}})+\alpha(1_{\mathcal{A}})b-log(det^{p}_{tr\otimes\alpha\otimes\beta}(p-A^{\dagger_{\mathcal{B}}}B^{\ddagger_{\mathcal{A}}}))\nonumber
\end{equation}
We say that $\de{a,b}$ are orthogonal, written $\de{a}\poll\de{b}$, when $\sca{a}{b}\neq 0,\infty$.
\end{definition}

\begin{definition}[Tensor Product]
Let $\de{a}=a\cdot+\cdot\alpha+A$ and $\de{b}=b\cdot+\cdot\beta+B$ be two projects of respective carriers $p,q$ with $pq=0$. The tensor product is defined as:
\begin{eqnarray}
\de{a\otimes b}&=&\sca{a}{b}\cdot +\cdot\alpha\otimes\beta+A^{\dagger_{\mathcal{B}}}+B^{\ddagger_{\mathcal{A}}}\nonumber\\
&=&(p+q,a\beta(1_{\mathcal{B}})+\alpha(1_{\mathcal{A}})b,\mathcal{A\otimes B},\alpha\otimes\beta,A^{\dagger_{\mathcal{B}}}+B^{\ddagger_{\mathcal{A}}})\nonumber
\end{eqnarray}
\end{definition}

\begin{definition}[Cut]
Let $\de{a}=a\cdot+\cdot\alpha+A$ and $\de{b}=b\cdot+\cdot\beta+B$ be two projects of respective carriers $p+q,q+r$ with $pq=qr=pr=0$. The cut is defined, when the feedback equation involving $A^{\dagger_{\mathcal{B}}}$ and $B^{\ddagger_{\mathcal{A}}}$ has a solution (denoted by $A^{\dagger_{\mathcal{B}}}\plug B^{\ddagger_{\mathcal{A}}}$), by:
\begin{eqnarray}
\de{a\plug b}&=&\sca{a}{b}\cdot +\cdot\alpha\otimes\beta+A^{\dagger_{\mathcal{B}}}\plug B^{\ddagger_{\mathcal{A}}}\nonumber\\
&=&(p+r,\sca{a}{b},\mathcal{A\otimes B},\alpha\otimes\beta,A^{\dagger_{\mathcal{B}}}\plug B^{\ddagger_{\mathcal{A}}})\nonumber
\end{eqnarray}
\end{definition}

\subsection{Embedding the graphs}

In this last part of the paper, we will associate to an operator graph an operator in the hyperfinite factor of type $\text{II}_{\infty}$. This allows us to associate to operator projects (projects whose graph is an operator graph) a Girard's project. We will show that this embedding preserves the measurement between projects, giving a combinatorial interpretation to Girard's measurement based on Fuglede-Kadison determinant. Moreover this embedding preserves both the tensor product and cut operations, so the interaction graphs can be seen as a combinatorial approach to Girard's GoI5.

As the reader will notice, the embedding comes down to a simple embedding of graphs in $\mathcal{B}(\hil{H})$. Thus we do not use the "type \text{II}" part of the hyperfinite factor. This can be explained very easily. As long as multiplicatives (and additives) are concerned, the use of a von Neumann algebra other than $\mathcal{B}(\hil{H})$ does not change much. The special features of type $\text{II}$ factors play a role in Girard's setting when dealing with second order quantification (though we believe a combinatorial approach would also work in this case) and with exponentiation.

Moreover, dialects and diatraces will not play a role in this paper, since they are not important when dealing with multiplicatives\footnote{They are important when dealing with additives and to define the contraction rule.}. Thus, the dialects of the Girard's project we will obtain through our embedding will all be equal to $\mathbb{C}$, and the diatrace will be the identity on $\mathbb{C}$ — denoted $1_{\mathbb{C}}$.

Following Girard, we will consider a trace $tr$ on $\infhyp $ given once and for all. For this reason, if $p$ is a finite projection the induced trace on $p\infhyp p$ is not normalized since $tr(1_{p\infhyp p})=tr(p)$. We will therefore denote (abusively) by $det_{FK}$ any extension of the usual Fuglede-Kadison determinant $\Delta$ on $p\infhyp p$ at the power $tr(p)$, a choice that is explained by the following remark.

\begin{remark}
Let $tr$ denote our fixed trace, $\lambda=tr(p)$, and let $T=tr/\lambda$ denote the normalized trace. Then for all invertible operator $A\in p\infhyp p$,
\begin{equation}
\Delta(A)^{\lambda}=exp(\lambda T(log(\abs{A})))=exp(tr(log(\abs{A})))\nonumber
\end{equation}
Hence the Fuglede-Kadison determinant raised to the power $\lambda$ corresponds to the “determinant" defined as in Definition \ref{FKdet} with a non-normalized trace such that $tr(1)=\lambda$ instead of the normalized trace $T$.
\end{remark}

\begin{proposition}\label{tracepreserving}
Let $\xi$ be a trace-preserving $\ast$-morphism from $\mathcal{M}_{n}(\mathbb{C})$ to $\infhyp $, and $u$ a matrix such that $\norm{u}\leqslant 1$, then 
\begin{equation}
det_{FK}(\xi(1-u))=\abs{det(1-u)}\nonumber
\end{equation}
\end{proposition}

\begin{proof}
Let $B_{1}$ be the unit ball of $\mathbb{C}$ and first suppose that $\text{Spec}_{\mathcal{M}_{n}(\mathbb{C})}(u)\subset B_{1}-\{1\}$. Then $\xi(u)$ satisfies $\text{Spec}_{\infhyp }(\xi(u))\subset B_{1}-\{1\}$ since the spectrum of $\xi(u)$ is contained in the spectrum of $u$. Moreover, $\xi$ is a $\ast$-homomorphism, and therefore commutes with the functional calculus, which means it commutes with the logarithm and the square root. Hence 
\begin{eqnarray}
det_{FK}(1-\xi(u))&=&exp(tr(log(|1-\xi(u)|)))\nonumber\\
&=&exp(tr(\xi(log(|1-u]))))\nonumber\\
&=&exp(tr(log(|1-u|)))\nonumber\\
&=&det(|1-u|)\nonumber\\
&=&det(((1-u)^{\ast}(1-u))^{\frac{1}{2}})\nonumber\\
&=&(det((1-u)^{\ast}(1-u)))^{\frac{1}{2}}\nonumber\\
&=&\abs{det(1-u)}\nonumber
\end{eqnarray}

Now, if $1\in\text{Spec}_{\mathcal{M}_{n}(\mathbb{C})}(u)$, then $1\in \text{Spec}_{\infhyp }(\xi(u))$ and the operators $u$ and $\xi(u)$ both have a nullspace, hence satisfy $det_{FK}(1-\xi(u))=0=\abs{det(1-u)}$ (using Lemma \ref{FKnullspace} for the left-hand equality).
\end{proof}

We now define the embedding, on operator projects.
\begin{definition}[Operator Project]
An operator project is a project $\de{a}=(a,A)$ where $A$ is an operator graph.
\end{definition}

From now on, we will write the hyperfinite factor $\infhyp $ of type $\text{II}_{\infty}$ as $\mathcal{B}(\hil{H})\vntimes\finhyp$, where $\finhyp$ denotes the hyperfinite factor of type $\text{II}_{1}$. We moreover consider the trace $tr$ defined as the tensor product of the normalized traces on $\finhyp$ and $\mathcal{B}(\hil{H})$. 

Let us denote by $\Phi$ the $\ast$-morphism $\mathcal{B}(\hil{H})\rightarrow \mathcal{B}(\hil{H})\otimes\finhyp$ defined as $a\mapsto a\otimes 1_{\finhyp}$. We associate to each operator project $\de{a}=(a,A)$ a Girard's project $\Phi(\de{a})=a\cdot+\cdot 1_{\mathbb{C}} + \Phi(\mat{A})$ of carrier $\Phi(p_{V_{A}})$, with $\mat{A}$ seen as an operator in $p_{V_{A}}\mathcal{B}(\hil{H})p_{V_{A}}$.

\begin{theorem}\label{meashyp}
The embedding preserves orthogonality and measurement, i.e. for any operator projects $\de{a}$ and $\de{b}$, we have $\sca{a}{b}=\mathopen{\ll}\Phi(\de{a}),\Phi(\de{b})\mathclose{\gg}$. Moreover, $\de{a}\poll\de{b}\Leftrightarrow \Phi(\de{a})\poll\Phi(\de{b})$.
\end{theorem}

\begin{proof}
The map $\Phi$ is obviously a trace-preserving injective $\ast$-morphism, hence its restrictions to $p \hil{H}p$, where $p$ is a finite projection, satisfy the hypotheses of Proposition \ref{tracepreserving}. From the facts that $\norm{\mat{A}\mat{B}}\leqslant 1$ and that $\mat{A}\mat{B}$ is a real matrix, we have that $\abs{det(1-\mat{A}\mat{B})}=det(1-\mat{A}\mat{B})$. Then, by Corollary \ref{ldet} and Proposition \ref{tracepreserving}, we obtain: 
\begin{eqnarray}
\sca{a}{b}&=&-log(det(1-\mat{A}\mat{B}))\nonumber\\
&=&-log(det_{FK}(1-(\mat{A}\otimes 1_{\finhyp})(\mat{B}\otimes 1_{\finhyp})))\nonumber\\
&=&\mathopen{\ll}\Phi(\de{a}),\Phi(\de{b})\mathclose{\gg}\nonumber
\end{eqnarray}

It immediately follows that $\de{a}\poll\de{b}$ if and only if $\Phi(\de{a})\poll\Phi(\de{b})$.
\end{proof}

\begin{theorem}
Let $\de{a}$ and $\de{b}$ be operator projects with disjoint carriers. Then 
\begin{equation}
\Phi(\de{a}\otimes\de{b})=\Phi(\de{a})\otimes\Phi(\de{b})\nonumber
\end{equation}
\end{theorem}

\begin{proof}
It is immediate that, when $A,B$ are simple graphs on disjoint sets of vertices, $\mat{A\cup B}$ is equal to $\mat{A}\oplus\mat{B}$. Since $\mat{A}$ and $\mat{B}$ are considered as operators in $p_{V_{A}}\mathcal{B}(\hil{H})p_{V_{A}}$ and $p_{V_{B}}\mathcal{B}(\hil{H})p_{V_{B}}$, their direct sum, as an operator in $(p_{V_{A}}+p_{V_{B}})\mathcal{B}(\hil{H})(p_{V_{A}}+p_{V_{B}})$ is equal to $\mat{A}+\mat{B}$. Hence 
\begin{eqnarray}
\Phi(a+b,A\cup B)&=&a+b\cdot+\cdot 1_{\mathbb{C}}+\mat{A\cup B}\otimes 1_{\finhyp}\nonumber\\
&=&a+b\cdot+\cdot 1_{\mathbb{C}}+(\mat{A}+\mat{B})\otimes 1_{\finhyp}\nonumber\\
&=&a+b\cdot+\cdot 1_{\mathbb{C}}+\mat{A}\otimes 1_{\finhyp}+\mat{B}\otimes 1_{\finhyp}\nonumber\\
&=&\Phi(a,A)\otimes\Phi(b,B)\nonumber\qedhere
\end{eqnarray}
\end{proof}

\begin{theorem}
Let $\de{a}$ and $\de{b}$ be operator projects, with $\sca{a}{b}\neq\infty$ (i.e. $\de{a}\plug\de{b}$ is defined). Then $\Phi(\de{a}\plug \de{b})=\Phi(\de{a})\plug\Phi(\de{b})$.
\end{theorem}

\begin{proof}
Let $(f,F)=\de{a}\plug\de{b}$. We showed that $\mat{F}$ is solution to the feedback equation between $\mat{A}$ and $\mat{B}$ (Proposition \ref{execution}). It is then clear that $\mat{F}\otimes 1_{\finhyp}$ is solution to the feedback equation between $\mat{A}\otimes 1_{\finhyp}$ and $\mat{B}\otimes 1_{\finhyp}$. Therefore $\mat{F}\otimes 1_{\finhyp}=\mat{A}\otimes 1_{\finhyp}\plug \mat{B}\otimes 1_{\finhyp}$. Moreover, we showed that $\sca{a}{b}=\mathopen{\ll}\Phi(\de{a}),\Phi(\de{b})\mathclose{\gg}$. Hence $\Phi((f,F))=\Phi(\de{a})\plug\Phi(\de{b})$.
\end{proof}

The last three theorems show how our framework can be regarded as a combinatorial approach to the operator algebraic construction of Girard \cite{goi5}. However, nothing insures us that our notion of success is preserved, and some hard work is required for that. In order to have a success-preserving embedding, we would need to construct a more explicit embedding of $\mathcal{B}(\hil{H})$ into the hyperfinite factor $\infhyp $ of type $\text{II}_{\infty}$ by means of operators obtained by pre-composition with measure-preserving maps. Such a construction would be very involved, and we believe that it extends beyond the scope of this paper.

\section{Conclusion}

We have shown how we can define a localized semantics where objects are graphs which yields a denotational semantics and a notion of truth. Eventually, we showed how we can reformulate all of our notions (when restricted to a certain class of graphs) in linear algebraic terms, which corresponds to GoI5 definitions exactly. Thus, we can consider our geometry of interaction as a combinatorial approach to GoI5. However, there is still some work to be done, and we develop below some directions that seem most interesting.

\paragraph{Extending our interaction graphs to additive and exponential connectives.} Our construction is very close to the construction that appears in Girard's last paper \cite{goi5}, but one could argue that its simplicity is due to its restriction to multiplicative connectives. We are however convinced that we will be able to construct all the required tools for the construction of additive and exponential connectives. Two directions seem to be of interest. The first one consists in extending our notion of graph by considering edge-colored graphs — which seems very promising, and the second would be to extend the set of possible weights to matrices — which is very close to Girard's use of \emph{dialects}.

\paragraph{An embedding in the hyperfinite factor that preserves truth.} In order to obtain an equivalence between the notions of success in our framework and Girard's geometry of interaction, it is necessary to use partial isometries that are the image (w.r.t. the viewpoint) of a partial measure-preserving bijection of $\mathbb{R}$ with the Lebesgue measure. It seems that we can give explicitly the projections and partial isometries used to represent matrices in the hyperfinite factor. Such a construction would therefore contain a rather technical explicit construction of the hyperfinite factor as a crossed product which we believe extended beyond the scope of this paper.

\paragraph{Extending our interaction graphs to allow sums of projects.} As we pointed out in the first section, the results obtained in sections \ref{goi} and \ref{denot} do not depend upon the choice of the function $m:]0,1]\rightarrow \mathbb{R}_{\geqslant 0}\cup\{\infty\}$ used to measure $1$-cycles. Indeed, the choice of the function $m(x)=-log(1-x)$ we used in this paper is needed for the sole purpose of the embedding into Girard's setting. In particular, it seems that this variability in the definitions could be interesting were we to look for a way to adapt our model to differential linear logic by allowing non-deterministic sums of projects.

\bibliographystyle{alpha}
\bibliography{goimult}

\appendix
\section{On von Neumann Algebras}\label{appendix}

This appendix is a survey of important results in the theory of operator algebras that we thought would help the reader grasp some of the technical results used in sections \ref{matrix} and \ref{hyp}. In the following, we will suppose familiarity with the notions of linear maps, Banach spaces, and Hilbert spaces. The reader interested in learning more of the theory can read the book of Murphy \cite{murphy} on the theory of C$^{\ast}$-algebras, the classical books of Kadison and Ringrose \cite{kadison1,kadison2}, or the more recent and quite complete series of Takesaki \cite{takesaki1,takesaki2,takesaki3}.

We here present some results of \emph{separable} operator algebras. We will consider given a Hilbert space $\hil{H}$ of infinite denumerable dimension, together with its inner product $\<\cdot,\cdot\>$ and the associated norm $\norm{\cdot}$, and develop the theory of operators (and algebras of operators) on $\hil{H}$. In particular, the notions we define are that of \emph{separable} concrete C$^{\ast}$-algebra, and \emph{separable} von Neumann algebra, even if it will never be explicit in the text.

\paragraph{Operators and Adjoints}~\\

We recall that an operator $T$ is a linear map from $\hil{H}$ to $\hil{H}$ that is continuous. The set of operators on $\hil{H}$ is denoted by $\mathcal{B}(\hil{H})$. A standard result tells us that $T$ being continuous is equivalent to $T$ being bounded, i.e. that there exists a constant $C$ such that for all $\xi\in\hil{H}$, $\norm{T\xi}\leqslant C\norm{\xi}$. The smallest such constant defines a norm on $\mathcal{B}(\hil{H})$ which we will denote by $\norm{T}$.

Being given an operator $T$ in $\mathcal{B}(\hil{H})$, we can show the existence of its \emph{adjoint} — denoted by $T^{\ast}$, the operator that satisfies $\< T\xi,\eta\>=\<\xi,T^{\ast}\eta\>$ for all $\xi,\eta\in\hil{H}$. It is easily shown that $T^{\ast\ast}=T$, i.e. that $(\cdot)^{\ast}$ is an involution, and that is satisfies the following conditions:
\begin{enumerate}
\item For all $\lambda\in\mathbb{C}$ and $T\in\mathcal{B}(\hil{H})$, $(\lambda T)^{\ast}=\bar{\lambda}T^{\ast}$;
\item For all $S,T\in\mathcal{B}(\hil{H})$, $(S+T)^{\ast}=S^{\ast}+T^{\ast}$;
\item For all $S,T\in\mathcal{B}(\hil{H})$, $(ST)^{\ast}=T^{\ast}S^{\ast}$.
\end{enumerate}

\paragraph{Topologies}~\\

In a Hilbert space $\hil{H}$ there are two natural topologies, the topology induced by the norm on $\hil{H}$, and a weaker topology defined by the inner product.
\begin{enumerate}
\item The strong topology: we say a sequence $\{\xi_{i}\}_{i\in\mathbb{N}}$ converges strongly to $0$ when $\norm{\xi_{i}}\rightarrow 0$.
\item The weak topology: a sequence $\{\xi_{i}\}_{i\in\mathbb{N}}$ converges weakly to $0$ when $\<\xi_{i},\eta\>\rightarrow 0$ for all $\eta\in\mathcal{B}(\hil{H})$. Weak convergence is thus a point-wise or direction-wise convergence.
\end{enumerate}

On $\mathcal{B}(\hil{H})$, numerous topologies can be defined, each of which having its own advantages and disadvantages. The five most important topologies are the norm topology, the strong operator topology, the weak operator topology, the ultra-strong (or $\sigma$-strong) topology and the ultra-weak (or $\sigma$-weak) topology. We can easily characterize the first three topologies in terms of converging sequences as follows :
\begin{enumerate}
\item The norm topology: $\{T_{i}\}_{i\in\mathbb{N}}$ converges (for the norm) to $0$ when $\norm{T_{i}}\rightarrow 0$;
\item The strong operator topology (SOT), which is induced by the strong topology on $\hil{H}$: $\{T_{i}\}_{i\in\mathbb{N}}$ converges strongly to $0$ when, for any $\xi\in\hil{H}$, $T_{i}\xi$ converges strongly to $0$;
\item The weak operator topology (WOT), which is induced by the weak topology on $\hil{H}$: $\{T_{i}\}_{i\in\mathbb{N}}$ converges weakly to $0$ when, for any $\xi\in\hil{H}$, $T_{i}\xi$ converges weakly to $0$.
\end{enumerate}

We won't however give the definitions of the ultra-strong and ultra-weak topologies here, and refer the interested reader to any standard textbook.

\paragraph{C$^{\ast}$-algebra and the continuous functional calculus}~\\

Notice that $\mathcal{B}(\hil{H})$, together with addition, composition and scalar multiplication has an algebra structure. An algebra possessing an involution $(\cdot)^{\ast}$ satisfying conditions (1)-(3) is called a $\ast$-algebra, or \emph{involutive algebra}. A $\ast$-subalgebra will therefore be a subalgebra of an involutive algebra which is closed under the involution. A $\ast$-morphism $\phi$ is an algebra morphism that satisfies $\phi(a^{\ast})=\phi(a)^{\ast}$.

\begin{Adefinition}[C$^{\ast}$-algebras]
A (concrete) C$^{\ast}$-algebra is a $\ast$-subalgebra of $\mathcal{B}(\hil{H})$ which is  norm-closed.
\end{Adefinition}

\begin{remark}
The abstract definition of a C$^{\ast}$-algebra says that $A$ is a C$^{\ast}$-algebra when $A$ is a Banach algebra satisfying $\norm{a}=\norm{a^{\ast}}$ and $\norm{a^{\ast}a}=\norm{a}^{2}$. The construction of Gelf'and-Naimark-Segal shows that any (abstract) C$^{\ast}$-algebra can be represented as a concrete C$^{\ast}$-algebra on a suitable Hilbert space.
\end{remark}

\begin{Adefinition}[Spectrum]
Let $T$ be an operator in a C$^{\ast}$-algebra $A$. We define its spectrum by $\text{Spec}_{A}(T)=\{\lambda\in \mathbb{C}~|~T-\lambda.1\text{ is not invertible in $A$}\}$.
\end{Adefinition}

\begin{remark}
In finite dimension, the spectrum of an operator (which can be written as a matrix in a given basis) is just the set of its eigenvalues.
\end{remark}

\begin{remark}
Let $X$ be a compact Hausdorff space. Let us write $C(X)$ for the set of continuous functions from $X$ to $\mathbb{C}$. Then it is a commutative C$^{\ast}$-algebra when considered with complex scalar pointwise multiplication and addition of functions, and and where $(\cdot)^{\ast}$ is defined by $f^{\ast}(x)=\overline{f(x)}$ ($\overline{(\cdot)}$ denotes complex conjugation). In this case, the spectrum of a function $f$ is its \emph{image}, i.e. the set $f(X)$.
\end{remark}

\begin{remark}
The spectrum of an operator $a$ in a unital Banach algebra $A$ is a non-empty closed subset of the disc of radius $\norm{a}$ centered on $0$ in the complex plane.
\end{remark}

\begin{Atheorem}[Continuous functional calculus]\label{spectral}
Let $a$ be a normal (i.e. $aa^{\ast}=a^{\ast}a$) element of a unital C$^{\ast}$-algebra $A$, and let $z$ be the inclusion map of $\text{Spec}_{A}(a)$ in $\mathbb{C}$. Then there exists a unique unital (i.e. $\phi(1)=1$) $\ast$-homomorphism $\phi : C(\text{Spec}_{A}(a))\rightarrow A$ such that $\phi(z)=a$.
\end{Atheorem}

\begin{remark}
If $f$ is a continuous function on $\text{Spec}_{A}(a)$, we define $f(a)$ as the operator $\phi(f)\in A$. It satisfies $\text{Spec}_{A}(f(a))=f(\text{Spec}_{A}(a))$.
\end{remark}

\paragraph{von Neumann Algebras}~\\

\begin{Adefinition}[von Neumann algebras]
A von Neumann algebra is a SOT-closed $\ast$-subalgebra of $\mathcal{B}(\hil{H})$.
\end{Adefinition}

Let $M\subset \mathcal{B}(\hil{H})$. We define the commutant of $M$ to be the set $M'=\{x\in\mathcal{B}(\hil{H})~|~\forall m\in M, mx=xm\}$. We will denote by $M''$ the double commutant $(M')'$ of $M$.

\begin{Atheorem}[von Neumann double commutation theorem]
Let $M$ be a $\ast$-subalgebra of $\mathcal{B}(\hil{H})$ with $1_{\hil{H}}\in M$. Then $M$ is a von Neumann algebra if and only if $M=M''$.
\end{Atheorem}

\begin{remark} Since the strong operator topology is weaker than the norm topology, a von Neumann algebra $M$ is also norm closed, hence a C$^{\ast}$-algebra. Moreover, since $M$ is the commutant of a set of operators, it contains the identity operator of $\mathcal{B}(\hil{H})$, hence it is a unital C$^{\ast}$-algebra. Therefore, we can define the continuous functional calculus for any normal operator of $M$.
\end{remark}

\paragraph{Factors and Types}~\\

Let $M$ be a von Neumann algebra. We define the \emph{center} of $M$ to be the von Neumann algebra $\mathcal{Z}(M)=M\cap M'$. 
\begin{Adefinition}[Factor]
A von Neumann algebra $M$ is called a \emph{factor} when its center is trivial, i.e. when $\mathcal{Z}(M)=\mathbb{C}.1_{M}$.
\end{Adefinition}

There exists a classification of factors based on the study of the set of projections and operators (partial isometries) between them. We recall that a projection is an operator $p$ satisfying $p^{2}=p=p^{\ast}$. If $M$ is a von Neumann algebra, we will write $\Pi(M)$ the set of all projections in $M$. We say two projections $p,q$ are disjoint when $pq=0$. It is standard that we can define a partial order on the set of projection $\Pi(\mathcal{B}(\hil{H}))$ by saying that $p\preceq q$ when $pq=p$. If $M$ is a von Neumann algebra, the restriction of this partial order to $M$ is obviously a partial order on $\Pi(M)$. 

An operator $u$ such that $u^{\ast}u$ is a projection (or equivalently, $uu^{\ast}$ is a projection) is called a partial isometry. In a von Neumann algebra $M$, we can define an equivalence relation on $\Pi(M)$ by saying that $p\sim_{M} q$ when there exists a partial isometry $u\in M$ such that $uu^{\ast}=p$ and $u^{\ast}u=q$. 

The partial order on the set of projections gives rise to a partial order $\precsim_{M}$ on the equivalence classes of projections, i.e. on $\Pi(M)/\sim$.

\begin{remark}
The fact that $p\preceq q$ means that $p\hil{H}$ is a subspace of $q\hil{H}$. The fact that $p\sim_{M} q$ means that $p\hil{H}$ and $q\hil{H}$ are \emph{internally (w.r.t. $M$) isomorphic}, i.e. they are isomorphic through an isomorphism $\phi$ which is an element of $M$. Therefore, the fact that $p\precsim_{M} q$ means that $p\hil{H}$ is \emph{internally} isomorphic to a subspace of $q\hil{H}$, hence that $p\hil{H}$ is somewhat \emph{internally (w.r.t $M$) smaller} than $q\hil{H}$.
\end{remark}

\begin{Adefinition}
We say a projection $p$ in a von Neumann algebra $M$ is infinite (in $M$) when there exists $q\prec p$ (i.e. a proper subprojection) such that $q\sim_{M} p$. A projection which is not infinite is called finite.
\end{Adefinition}

\begin{Aproposition}
Let $M$ be a von Neumann algebra. Then $M$ is a factor if and only if the relation $\precsim_{M}$ is a total order.
\end{Aproposition}

In the statement of the following theorem, we use the usual notion of order type with the exception that we make a difference between $\infty$ and any other element, considering that $\infty$ represents a class of infinite projections. For instance, $\{0,1\}$ and $\{0,\infty\}$ should be considered as different order types since the first contains no infinite element, while the second do.

\begin{Aproposition}[Type of a factor]
Let $M$ be a factor. We say that:
\begin{itemize}
\item $M$ is of type $\text{I}_{n}$ when $\precsim_{M}$ is of the same order type as $\{0,1,\dots,n\}$;
\item $M$ is of type $\text{I}_{\infty}$ when $\precsim_{M}$ is of the same order type as $\mathbb{N}\cup\{\infty\}$;
\item $M$ is of type $\text{II}_{1}$ when $\precsim_{M}$ is of the same order type as $[0,1]$;
\item $M$ is of type $\text{II}_{\infty}$  when $\precsim_{M}$ is of the same order type as $\mathbb{R}_{\geqslant 0}\cup\{\infty\}$;
\item $M$ is of type $\text{III}$  when $\precsim_{M}$ is of the same order type as $\{0,\infty\}$, i.e. all non-zero projections are infinite.
\end{itemize}
Moreover, $\precsim_{M}$ cannot be of any other order type.
\end{Aproposition}

It can be shown that a type $\text{I}_{n}$ factor is isomorphic to $\mathcal{M}_{n}(\mathbb{C})$, the algebra of complex $n\times n$ matrices. A type $\text{I}_{\infty}$ factor is isomorphic to $\mathcal{B}(\hil{H})$.

\begin{remark}
Following the preceding remark, when restricting to a subalgebra of $\mathcal{B}(\hil{H})$, we lose some operators, and in particular some partial isometries.  There is here an obvious analogy to make with Skolem's paradox, where we can find non-denumerable sets in a denumerable model of set theory. Therefore, a set $X$ is \emph{internally} non-denumerable because there are no maps \emph{inside the model} from $\omega$ to $X$, but it is externally denumerable because once out of the model, one will find a suitable map. Here, the same thing happens: the subspaces $p\hil{H}$ and $q\hil{H}$ can be isomorphic, but one would have to step \emph{outside} of $M$ to see it (i.e. there are no partial isometry from $p\hil{H}$ onto $q\hil{H}$ in $M$). So the subspaces are \emph{not isomorphic} from the point of view of $M$, even though they are from the point of view of $\mathcal{B}(\hil{H})$.
\end{remark}

\paragraph{Traces}~\\

\begin{Adefinition}
Let $a$ be a self-adjoint operator in $M$. We say that $a$ is \emph{positive} if $\text{Spec}_{M}(a)\subset \mathbb{R}_{\geqslant 0}$. We denote by $M^{+}$ the set of positive operators in $M$.
\end{Adefinition}

\begin{Aproposition}
We have $M^{+}=\{u^{\ast}u~|~u\in M\}$.
\end{Aproposition}

\begin{Adefinition}
A \emph{trace} $\tau$ on a von Neumann algebra $M$ is a function from $M^{+}$ in $[0,\infty]$ satisfying :
\begin{enumerate}
\item $\tau(x+y)=\tau(x)+\tau(y)$ for all $x,y\in M^{+}$.
\item $\tau(\lambda x)=\lambda\tau(x)$ for all $x\in M^{+}$ and $\lambda\geqslant 0$ 
\item $\tau(x^{\ast}x)=\tau(xx^{\ast})$ for all $x\in M$
\end{enumerate}
We say that it is \emph{faithful} if $\tau(x)>0$ for all $x\not= 0$ in $M^{+}$, that it is \emph{finite} when $\tau(1)<\infty$, that it is \emph{normal} when $\tau(\sup \{x_{i}\})=\sup \{\tau(x_{i})\}$ for any bounded increasing net $\{x_{i}\}$ in $M^{+}$.
\end{Adefinition}

\begin{Atheorem}
If $M$ is a finite factor (i.e. the identity is finite), it admits a finite faithful normal trace $\tau$. Moreover, any other finite faithful normal trace $\rho$ is proportional to $\tau$.

If $M$ is of type $\text{II}_{1}$, we call the \emph{normalized} trace the unique finite faithful normal trace $T$ such that $T(1)=1$.
\end{Atheorem}

\begin{remark} Since the positive operator in $M$ linearly span the von Neumann algebra $M$, a finite trace $\tau$ extends uniquely to a positive linear functional on $M$ which we will call $\tau$. In particular, one can define the trace of any operator $a$ in a type $\text{II}_{1}$ factor.
\end{remark}

\paragraph{Hyperfiniteness.}~\\

\begin{Adefinition}
We say that a von Neumann algebra $M$ is \emph{hyperfinite} (or \emph{approximately finite dimensional}) if there is a directed collection $M_{i}$ of ﬁnite-dimensional $\ast$-subalgebras of $M$ such that the union $\cup_{i}M_{i}$ is dense in $M$ for the ultra-weak topology. 
\end{Adefinition}

\begin{remark}
The hyperfiniteness of the factor $M$ should be thought of as the fact that the operators in $M$ can be approximated by matrices.
\end{remark}

\begin{Atheorem}
The hyperfinite factor $\finhyp$ of type $\text{II}_{1}$ is unique up to isomorphism. 
\end{Atheorem}

\begin{Atheorem}
The hyperfinite factor $\infhyp $ of type $\text{II}_{\infty}$ is unique up to isomorphism. In particular, it is isomorphic to the von Neumann algebra tensor product $\mathcal{B}(\hil{H})\otimes \finhyp$.
\end{Atheorem}

\end{document}